\documentclass[runningheads]{llncs}
\usepackage{graphicx}
\usepackage{amsfonts}
\usepackage{amsmath}
\usepackage{verbatim}
\usepackage{breqn}
\usepackage{subfig}
\usepackage[usenames,dvipsnames]{xcolor}
\usepackage[hidelinks]{hyperref}
\usepackage{lmodern}
\usepackage{enumitem}
\usepackage{environ}

\newcommand{\C}{\mathcal{C}}

\newcommand{\eps}{\varepsilon}
\newcommand{\vf}{V\setminus V(F)}

\newcommand{\I}{\mathcal{I}}

\newcommand{\J}{\mathcal{J}}

\newcommand{\GF}{G\setminus F}
\newcommand{\bb}[1]{\#^a #1}
\newcommand{\G}{\mathbb{G}}
\newcommand{\Gr}{\mathcal{G}\!r}

\newcommand{\repeattheorem}[1]{%
  \begingroup
  \renewcommand{\thetheorem}{\ref{#1}}%
  \expandafter\expandafter\expandafter\theorem
  \csname reptheorem@#1\endcsname
  \endtheorem
  \endgroup
}

\NewEnviron{reptheorem}[1]{%
  \global\expandafter\xdef\csname reptheorem@#1\endcsname{%
    \unexpanded\expandafter{\BODY}%
  }%
  \expandafter\theorem\BODY\unskip\label{#1}\endtheorem
}

\newcommand{\repeatlemma}[1]{%
  \begingroup
  \renewcommand{\thelemma}{\ref{#1}}%
  \expandafter\expandafter\expandafter\lemma
  \csname replemma@#1\endcsname
  \endlemma
  \endgroup
}

\NewEnviron{replemma}[1]{%
  \global\expandafter\xdef\csname replemma@#1\endcsname{%
    \unexpanded\expandafter{\BODY}%
  }%
  \expandafter\lemma\BODY\unskip\label{#1}\endlemma
}

\spnewtheorem*{proofsketch}{Sketch of Proof}{\itshape}{\rmfamily}
\if@envcntsame 
   \def\spn@wtheorem#1#2#3#4{\@spothm{#1}[theorem]{#2}{#3}{#4}}
\else 
   \if@envcntsect 
      \def\spn@wtheorem#1#2#3#4{\@spxnthm{#1}{#2}[section]{#3}{#4}}
   \else 
      \if@envcntreset
         \def\spn@wtheorem#1#2#3#4{\@spynthm{#1}{#2}{#3}{#4}
                                   \@addtoreset{#1}{section}}
      \else
         \def\spn@wtheorem#1#2#3#4{\@spynthm{#1}{#2}{#3}{#4}
                                   \@addtoreset{#1}{chapter}}%
      \fi
   \fi
\fi

\begin{document}
	\title{$\beta$-Stars or On Extending a Drawing of a Connected Subgraph }
	\author{Tamara Mchedlidze\inst{1} \and J\'er\^ome Urhausen\inst{2}}

	\authorrunning{T. Mchedlidze, J. Urhausen}

	\institute{Karlsruhe Institute of Technology, \email{mched@iti.uka.de}\and
Utrecht University, \email{j.e.urhausen@uu.nl}}

	\maketitle

	\begin{abstract}
We consider the problem of extending the drawing of a subgraph of a given plane graph to a drawing of the entire graph using straight-line and polyline edges. We define the notion of star complexity of a polygon and show that a drawing $\Gamma_H$ of an induced connected subgraph $H$ can be extended with at most $\min\{ h/2, \beta + \log_2(h) + 1\}$ bends per edge, where $\beta$ is the largest star complexity of a face of $\Gamma_H$ and $h$ is the size of the largest face of $H$. This result significantly improves the previously known upper bound of $72|V(H)|$~\cite{chan2015drawing} for the case where $H$ is connected.
We also show that our bound is worst case optimal up to a small additive constant.
Additionally, we provide an indication of complexity of the problem of testing whether a star-shaped inner face can be extended to a straight-line drawing of the graph; this is in contrast to the fact that the same problem is solvable in linear time for the case of  star-shaped outer face~\cite{hong2008convex} and convex inner face~\cite{mchedlidze2016extending}.
	\end{abstract}

\section{Introduction}
\label{ch:introduction}

In this paper we study the problem of extending a given partial drawing of a graph. In particular, given a plane graph $G=(V,E)$, i.e. a planar graph with a fixed combinatorial embedding and a fixed outer face, a subgraph $H$ of $G$ and a planar straight-line drawing $\Gamma_H$ of $H$, we ask whether $\Gamma_H$ can be extended to a planar straight-line drawing of $G$ (see Figure~\ref{fig:IntroInstance.pdf}). We study both the decision question and the relaxed variation of using bends for the drawing extension.

\begin{figure}[tbh]
	\centering
	\includegraphics[scale=0.8]{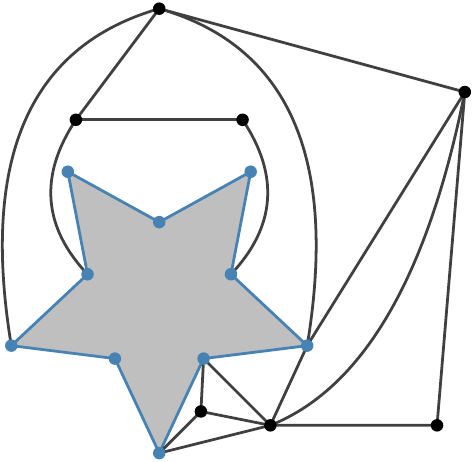}
	\caption{An embedding of a plane graph alongside a fixed drawing of an inner face (blue) as a star-shaped polygon (gray).}
	\label{fig:IntroInstance.pdf}
\end{figure} 

It is known that a drawing extension always exists even if $H=(V,\emptyset)$, where each edge is represented by a polyline with at most $120n$  bends, here $n=|V|$~\cite{pach2001embedding}. This bound was improved to $3n+2$ by Badent et al.~\cite{Badent08}. These upper bounds are asymptotically optimal as there are instances that require $\Omega(n)$ bends on $\Omega(n)$ edges~\cite{Badent08}.
In terms of the size of the pre-drawn graph $H$,  Chan et al.~\cite{chan2015drawing} showed that a drawing extension with $72|V(H)|$ bends per edge is possible for a general subgraph $H$. 

In order to pinpoint the source of multiple necessary bends for the drawing extension we define the notion of a $\beta$-star (resp.\ $\beta$-outer-star), a polygon where $\beta$ bends are necessary and sufficient to reach the kernel of the polygon (resp.\ infinity). We study the upper bounds on the number of bends in a drawing extension as a function of $\beta$. We show that a drawing $\Gamma_H$ of an induced connected subgraph $H$ can be extended with at most $\min\{ h/2, \beta + \log_2(h) + 1\}$ bends per edge if each face of $H$ is represented in $\Gamma_H$ as a $\beta$-(outer)-star and $h$ is the size of the largest face of $H$ (Theorem~\ref{th:ExtendArbitraryCycle}).  We show that this bound is worst case optimal up to a small additive constant.
We observe that in case both $G$ and $H$ are trees a closer to optimal bound of  $1+2\lceil|V(H)|/2\rceil)$ bends per edge had been provided by  Di Giacomo et al.~\cite{DiGiacomo09}.

In case a planar embedding is not provided as a part of the input, it is NP-hard to test whether a straight-line drawing extension exists~\cite{patrignani2006extending}. The problem is not known to belong to the class NP, as a possible solution may have coordinates which can not be represented with a polynomial number of bits~\cite{patrignani2006extending}. Very recently, Lubiw et al. have studied a related problem of drawing a graph inside a (not-necessarily simply connected) closed polygon~\cite{Lubiw2018}. They showed that this problem can not be shown to lie in NP by the mean of providing vertex coordinates, as these are sometimes irrational numbers. They have also shown that the problem is hard for the existential theory of  reals ($\exists \mathbb{R}$-hard) even if a planar embedding of the graph is provided as a part of the input. 
This problem would be equivalent to partial graph drawing extendability, if the polygon would be open, however this situation has not been investigated.  
Bekos et al.~\cite{BekosGD18,BekosArxiv17} have studied the problem of extending a given partial drawing of bipartite graphs, where one side of the bipartition is pre-drawn. They have shown that this problem lies in NP if each free vertex is required to lie in the convex hull of its pre-drawn neighbors.
Regarding drawing extensions with bends, it is NP-hard to test whether a drawing extension with at most $k$ bend per edge exists~\cite{BekosGD18,Goaoc2009}.  

Despite all the hardness results, it is long known that a straight-line drawing extension always exists if $H$ is the outer face and $\Gamma_H$ is a convex polygon~\cite{Chambers12,tutte1963draw}; and
$H$ is a chordless outer face and $\Gamma_H$ is a star-shaped polygon~\cite{hong2008convex}. An existence of a straight-line drawing extension can be checked by the mean of necessary and sufficient conditions in case where $H$ is an inner face and $\Gamma_H$ is a convex polygon~\cite{mchedlidze2016extending}.  As an extension of this work, and with the general goal to better understand the boundary between the easy and the difficult cases, we investigated the question of testing whether a straight-line drawing extension exists for an inner face $H$ drawn as a star-shaped polygon~$\Gamma_H$.
We observe that one can not test whether such an extension exists by just checking each vertex individually, as in the case for a convex inner face, and
show that there exists an instance such that the region where a vertex of $V(G)\setminus V(H)$ can lie to allow for a straight-line drawing extension is bounded by a curve of degree $2^{\Omega(|H|)}$(Theorem~\ref{th:ExpoComplBaseCase}).    
        
~\\
\noindent{\bf Contribution and Outline.}
We start with the necessary definitions in Section~\ref{ch:preliminaries}.
In Section~\ref{ch:StarOneBend}, we show that a star-shaped drawing of an inner face can be extended with at most 1 bend per edge. 
Section~\ref{ch:GeneralStars} is devoted to the study of generalizations of stars.
In Section~\ref{ch:beta-subsection}, we start with a generalization of star-shaped polygons to $\beta$-star and $\beta$-outer-star polygons ($\beta$ is referred to as  star complexity), and show that the number of bends per edge necessary for a drawing extension of an inner face $H$ with a star complexity $\beta$ is not bounded in terms of $\beta$ (Theorem~\ref{th:InnerStarAsInnerFaceLowerBound} and Theorem~\ref{th:lowerlog}). Motivated by the proof of Section~\ref{ch:StarOneBend} we define the notion of planar-$\beta$-star and planar-$\beta$-outer-star (this $\beta$ is referred to as planar star complexity) and show that the planar star complexity determines the number of bends per edge in a drawing extension (Theorem~\ref{th:ExtendPlanarInnerStar} and Theorem~\ref{th:ExtendPlanarOuterStar}).  In Section~\ref{ch:resolve-subsection}, we study the planar star complexity of an arbitrary simple polygon and the relationship between the star complexity and the planar star complexity of a polygon. In particular, we show that every $\beta$-star with $n$ vertices is a planar-$\beta\!+\!\delta$-star where $\delta\leq \log_2(n)$ (Theorem~\ref{th:OuterToPlanar}). 
In Section~\ref{ch:GeneralGraphs}, we state the implications of Section~\ref{ch:GeneralStars} to the drawing extension of (induced) connected subgraphs. In particular, we prove that a drawing $\Gamma_H$ of an induced connected subgraph $H$ can be extended with at most $\min\{ h/2, \beta + \log_2(h) + 1\}$ bends per edge if the star complexity of $\Gamma_H$ is $\beta$ and $h$ is the size of the largest face of~$H$ (Theorem~\ref{th:ExtendArbitraryCycle}). Last but not least, in Section~\ref{ch:NoBends} we provide an indication of complexity of the problem of testing whether a star-shaped inner face $H$ admits a straight-line drawing extension. In particular, we prove that there exists an instance such that the region where a vertex of $V(G)\setminus V(H)$ can lie to allow a straight-line drawing extension is bounded by a curve of degree $2^{\Omega(|H|)}$ (Theorem~\ref{th:ExpoComplBaseCase}).
All omitted proofs can be found in the appendix.

\section{Preliminaries}
\label{ch:preliminaries}


\paragraph{Basic Geometric Terms.}
The segment (resp.\ line) induced by two points $a$ and $b$ is designated by $s(a,b)$ (resp.\ $l(a,b)$). We denote a curve between $a$ and $b$ by $c(a, b)$. We refer to the ray along $l(a,b)$ starting at $a$ and (not) containing $b$ as $r(a,b)$  ($q(a,b)$). 
For a polyline $c$, $\#c$ designates the number of bends on $c$.

Let $P$ be a polygon. Two points $a,b$ \emph{see} each other if the open segment $s(a,b)$ does not intersect the boundary of $P$. A simple polygon $P$ is \emph{convex} if each pair of points inside $P$ see each other.
A simple polygon $P$ is \emph{star-shaped} or a \emph{star} if there is a non-empty set of points $K$ called the \emph{kernel} inside the polygon such that
any point of the kernel can see any vertex of the polygon. By assuming that the vertices of $P$ are in general position, we have that a kernel of $P$  contains an open ball of positive radius.

\paragraph{Graphs and Drawings of Graphs.}
A \emph{drawing} $\Gamma$ of a graph is a function that assigns to each vertex a unique point in the plane and to each edge $\{a,b\}$ a curve connecting the points assigned to $a$ and $b$. A drawing is \emph{straight-line} (resp.\ \emph{$k$-bend}) if each edge is drawn as a segment (resp.\ a polyline with at most $k$ bends).
A graph  is \emph{planar} if it has a \emph{planar} drawing, i.e. a drawing without edge crossings. A planar drawing $\Gamma$
subdivides the plane into connected regions called \emph{faces}; the
unbounded region is the \emph{outer} and the other regions are the
\emph{inner} faces. The cyclic
ordering of the edges around each vertex of $\Gamma$ together with the
description of the outer face of $\Gamma$ characterize a class of drawings with the same combinatorial properties, which is called an \emph{embedding}
of $G$. A planar graph $G$ with a planar embedding is called \emph{plane graph}. A \emph{plane subgraph} $H$ of $G$ is a subgraph of $G$
together with a planar embedding that is the restriction of the
embedding of $G$ to~$H$.
A plane graph $G$ is (\emph{internally}) \emph{triangulated} if each (inner) face of $G$ is a triangle.
For a given cycle, a chord is an edge between two non-consecutive vertices of the cycle.

Let $G$ be a plane graph and let $H$ be a plane subgraph of 
$G$.  Let $\Gamma_H$ be a planar straight-line drawing of $H$.
We say that the instance $(G, \Gamma_H)$ admits a \emph{$k$-bend} (resp.\ \emph{straight-line}) \emph{extension} if drawing $\Gamma_H$ can be 
completed to a planar $k$-bend (resp.\ straight-line) drawing $\Gamma_G$ of the plane graph $G$.  We refer to $k$ as the \emph{curve complexity} of the drawing $\Gamma_G$.

For a given graph $G=(V,E)$, let $N(v)=\{w\in V\mid \{v,w\}\in E\}$ be the \emph{neighbors} of $v\in V$.
For a plane graph $G$ and a face $F$, let $N_F(v)=N(v)\cap F=(w_1,w_2,\dots,w_\ell)$ be the sequence of neighbors of $v$ that belong to $F$.
For $v$ outside $F$, let the list $N_F(v)$ be ordered clockwise around $F$ with $w_1$ chosen such that the area delimited by the cycle $C$ composed of edges $\{v, w_1\}$, $\{v, w_\ell\}$ and the clockwise path $H$ from $w_1$ to $w_\ell$ in $F$ does not contain $F$ (see Figure~\ref{fig:NeighborsInF-PutVertexOnRay}a).
A vertex $z\in \vf$ lying in the cycle $C$ is said to be \emph{enclosed} by vertex $v$.

Let $F$ be a face of $G$ and $\Gamma_F$ its planar drawing. 
The \emph{feasibility area} of a vertex $v\in \vf$ is the set of all possible positions of $v$, such that 
the implied straight-line drawing of $F\cup \{v\}$  can be extended to a planar straight-line drawing of $V(F)\cup \{v\}\cup Q_v$, where $Q_v$ is the set of all vertices enclosed by $v$.

	\section{Star-shaped polygons }\label{ch:StarOneBend}

Let $G$ be a plane graph with $n$ vertices, $F$ be a chordless face of $G$ with $h$ vertices and $\Gamma_F$ a star-shaped drawing of $F$.
In this section we prove that the instance $(G,\Gamma_F)$ admits a $1$-bend-extension. While the proof itself is rather straight-forward, we still present it here as it motivates a specific way to generalize star-shaped polygons by considering planarity issues. 

In our construction we place vertices $\vf$  one by one  with the property that a vertex is placed only after all vertices enclosed by it have  already been placed. This property is achieved by a canonical ordering~\cite{kant1996drawing} that lists vertices starting from the face $F$. The following lemma can be proven along the same lines as the existence of a usual canonical ordering~\cite{kant1996drawing}. We say $\GF$ is triangulated if each face of $G$ is triangulated with the exception of the face $F$.

\begin{lemma}\label{th:ordering}
Let $G=(V,E)$ be a plane graph, $|V|=n$, and let $F$ be an inner face with $h$ vertices of $G$, such that $\GF$ is triangulated.
There is an ordering $\J=(v_1,\dots,v_{n-h})$ of the vertices of $\vf$, such that for each $j$, $1\leq j\leq n-h$, the following holds:
$(1)$ the graph $G_j$ induced by the vertices $\{v_1,\dots,v_j\}\cup F$ is biconnected, $(2)$ $G_j\setminus F$ is internally triangulated,
$(3)$ $v_{j+1}$ lies in the outer face of $G_j$, $(4)$ vertices $N(v_{j+1}) \cap V(G_j)$ belong to the outer face of $G_j$.
\end{lemma}

\begin{theorem}\label{th:OneBend}
Each instance $(G=(V, E),\Gamma_F)$ where $\Gamma_F$ is a star-shaped drawing of a chordless inner face $F$  allows a $1$-bend-extension.
\end{theorem}
\begin{proof}
\begin{figure}[t]
\centering
\includegraphics[scale=0.8]{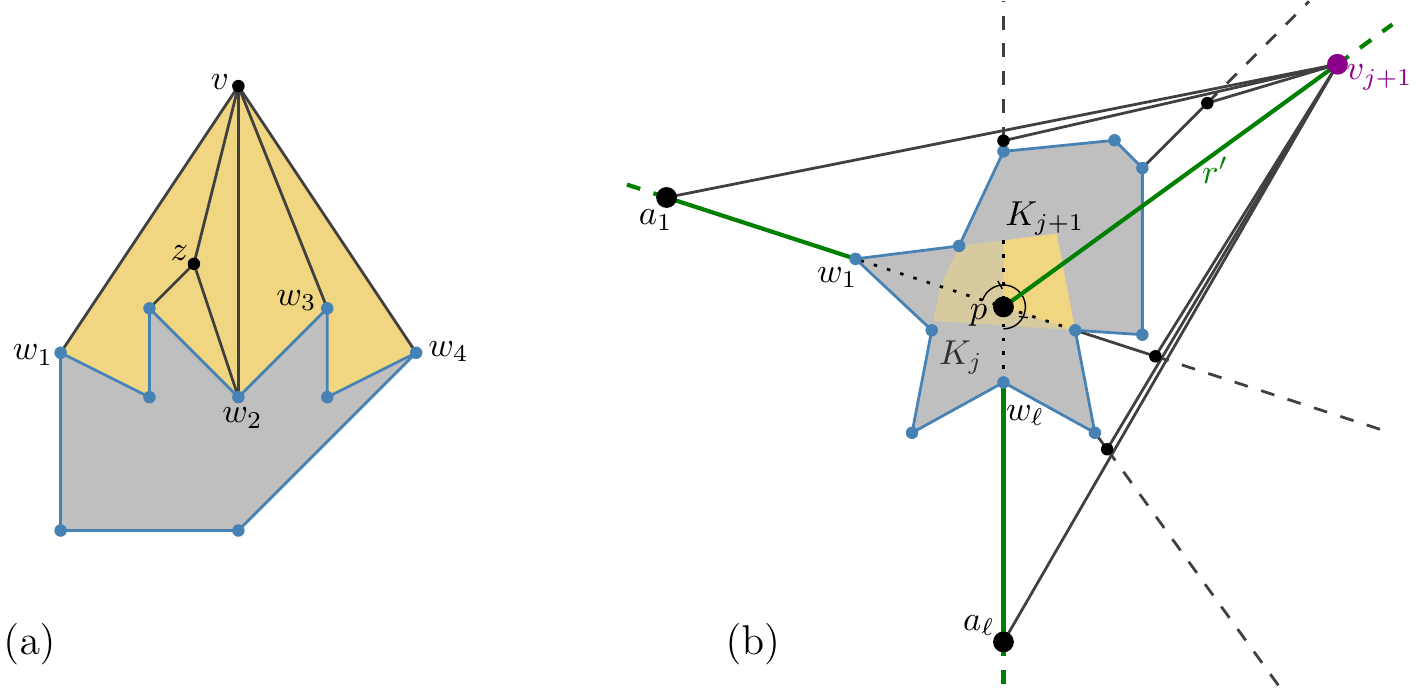}
\caption{ (a) Yellow area contains vertices enclosed by $v$. (b) Proof of Theorem~\ref{th:OneBend}.}
\label{fig:NeighborsInF-PutVertexOnRay}
\end{figure}

We start with triangulating $G$ by placing a vertex in each non-triangular face and connecting it to the vertices of the face. 
We delete the added vertices and edges after the triangulated graph has been drawn. We refer to the new graph as $G$ as well.
Let~$\J=(v_1,\dots,v_{n-h})$ be an ordering of the vertices $\vf$ as defined by Lemma~\ref{th:ordering}.
For $1\leq j\leq n-h=|V(G)|-|V(F)|$, let $G_j$ be the graph as defined by Lemma~\ref{th:ordering} and let $F_j$ be the outer face of $G_j$. Additionally we set $G_0=F_0=F$.

We prove the theorem by induction.
Assume that for a $0\leq j\leq n-h$ we have a drawing of $G_j$, such that $F_j$ forms a star-shaped polygon $P_j$ with kernel~$K_j$.
This is true for $j=0$. 
Let $v_{j+1}$ be the next vertex according to $\J$ and let $p$ be a point of the kernel of the already drawn star-shaped polygon $P_j$.
For each $w\in N_{F_j}(v_{j+1})$ consider the ray $q(p,w)$. Due to $P_j$ being star-shaped and due to property (4) of Lemma~\ref{th:ordering}, they all lie outside of $P_j$. Since $G_j$ is biconnected, $v_{j+1}$ has at least two neighbors, i.e. $\ell=|N_{F_j}(v_{j+1})|\geq 2$.

Now we consider the ray $r'$ that is the bisector of the clockwise angle formed by the rays $q(p,w_1)$ and $q(p,w_\ell)$, see Figure~\ref{fig:NeighborsInF-PutVertexOnRay}b.
If we place $v_{j+1}$ sufficiently far away from $p$ on $r'$, $v_{j+1}$ sees $q(p,w_1)$ and $q(p,w_\ell)$, i.e. $\exists a_1\in q(p,w_1), a_\ell\in q(p,w_\ell)$, with $s(v_{j+1},a_1)\cap P_j=\emptyset=s(v_{j+1},a_\ell)\cap P_j$. This is due to the fact that the angles between $q(p,w_1)$ and $r'$ and between $q(p,w_\ell)$ and $r'$ are strictly smaller than $\pi$.

Since $v_{j+1}$ is between $q(p,w_1)$ and $q(p,w_\ell)$, $v_{j+1}$ also sees a point $a_i$ on the ray $q(p,w_i)$, $i=2,\dots,\ell-1$.
For each $i\in \{1,\dots,\ell\}$ we draw the edge $\{w_i,v_{j+1}\}$ using the segments $s(w_i,a_i)$ and $s(a_i, v_{j+1})$. Observe that the points $a_1,\dots,a_\ell$ should be chosen so that they appear around $v_{j+1}$ in a counterclockwise order.

The lines $l(p,w_1)$ and $l(p,w_\ell)$ separate the plane into four quadrants. The new kernel $K_{j+1}$ of the polygon $P_{j+1}$ is the intersection of the old kernel $K_j$ and the quadrant containing $v_{j+1}$.
Since the kernel $K_j$ was an open set, $p$ could not have been on the boundary of $K_j$, therefore $K_{j+1}$ is a non-empty open set.\qed
\end{proof}

We observe that, according to the proof of Theorem~\ref{th:OneBend}, the class of the polygons that allows a $1$-bend-extension is wider than stars. In particular, these are the polygons from the vertices of which we can shoot rays to infinity which neither intersect mutually nor intersect the polygon itself. We call such polygons \emph{planar outer-stars}. This gives the following: 

\begin{corollary}\label{th:corollary-outer-one-bend}
	Each instance $(G,\Gamma_F)$ where  $F$ is a chordless inner face and $\Gamma_F$ is a planar outer-star,  allows a $1$-bend-extension. 
\end{corollary}

	\section{Generalization of stars}\label{ch:GeneralStars}

In this section we generalize the notion of stars and planar outer-star polygons and investigate the lower and upper bounds for the number of bends per edge in the drawing extensions. 

\subsection{$\beta$-Stars}
\label{ch:beta-subsection}
A simple polygon $P$ is a \emph{$\beta$-star} if there is an open set of points $K$ called the \emph{kernel} inside $P$ with the following property:
for each point $p\in K$ and for each vertex $v$ of $P$ there is a polyline $c(v)$ connecting $v$ and $p$ with at most $\beta$ bends such that $c(v)$ touches $P$ only at $v$. 
The smallest such $\beta$ is referred to as \emph{star complexity} of the polygon $P$.
This set of curves is referred to as \emph{curve-set} $\C$ of $P$ and $p$ is the \emph{center} of $\C$.
In the literature this kernel is also known as the link center of the polygon and it can be calculated in $O(n\log n)$ time~\cite{djidjev1992ano}.
The straight-forward extension of this definition to act ``outside'' the polygons is as follows: a simple polygon $P$ is a \emph{$\beta$-outer-star} if for each vertex $v$ of $P$ there is an infinite polyline $c(v)$ outside of $P$ starting at $v$ with at most $\beta$ bends. The smallest such $\beta$ is referred to as \emph{outer star complexity} of the polygon $P.$
Again, $\C=\{c(v)\mid v\in P\}$ is called \emph{curve-set}. The \emph{center} of this set is a point at infinity.
One can think about $\beta$-outer-star as of $\beta$-star with the kernel  in infinity.

While $\beta$-star and $\beta$-outer-star are straight-forward ways to extend the notion of a star inside and outside, and these definitions capture an inherent complexity of the polygon,   
we can show that restricting the fixed inner face to be a $1$-star is not sufficient to ensure a $c$-bend-extension for any constant $c$  (Theorem~\ref{th:InnerStarAsInnerFaceLowerBound}).
Even more, restricting the fixed inner face to a $\beta$-outer-star still does not imply the existence of a $c\!+\!\beta$-bend-extension for any constant $c$ (Theorem~\ref{th:lowerlog}).

\begin{reptheorem}{th:InnerStarAsInnerFaceLowerBound}
There exist instances $(G,\Gamma_F)$ where $F$ is an inner face with $h$ vertices and $\Gamma_F$ is a $1$-star such that any drawing extension of $(G,\Gamma_F)$ contains an edge with at least $\lfloor\frac{h-3}{2}\rfloor$ bends.
\end{reptheorem}

\begin{reptheorem}{th:lowerlog}
There exist instances $(G,\Gamma_F)$ where   $F$ is an inner face  with $h$ vertices and $\Gamma_F$ is a $\beta$-outer-star such that any drawing extension of $(G,\Gamma_F)$  has an edge with at least $\beta+ \log_2(\frac{h+5}{6})+1$ bends.
\end{reptheorem}

The above lower bounds and the fact that a planar outer-star admits an extension with one bend per edge guided us to extend definitions of $\beta$-star and $\beta$-outer star to include planarity.  
A simple polygon $P$ is a \emph{planar-$\beta$-star} if there is an open set of points $K$ called the kernel inside $P$ with the following property: for a fixed point $p\in K$ and for each vertex $v$ of $P$ there is an oriented polyline $c(v)$ inside $P$ from $v$ to $p$ with at most $\beta$ bends such that for any $v$ and $v'$, $c(v)$ and $c(v')$ share the single point $p$.
	
A simple polygon $P$ is a \emph{planar-$\beta$-outer-star} if for each vertex $v$ of $P$ there is an oriented infinite polyline $c(v)$ outside of $P$ starting at $v$ with at most $\beta$ bends such that for any $v$ and $v'$, $c(v)$ and $c(v')$ neither cross nor touch each other.
The smallest such $\beta$ is referred to as \emph{planar (outer) star complexity} of the polygon $P$.
The set of curves are referred to as \emph{planar curve-set}
centered at the fixed point $p$.
Figure~\ref{fig:DifferencePlanar} in the appendix shows that in general a $\beta$-outer-star is not a planar-$\beta$-outer-star.
Due to these definitions the following two theorems can be proven.



\begin{reptheorem}{th:ExtendPlanarInnerStar}
	Each instance  $(G,\Gamma_F)$ where $F$ is a chordless outer face and $\Gamma_F$ is a planar-$\beta$-star allows a $\beta$-bend-extension.
\end{reptheorem}

\begin{reptheorem}{th:ExtendPlanarOuterStar}
Each instance  $(G,\Gamma_F)$ where $F$ is a chordless inner face and $\Gamma_F$ is a planar-$\beta$-outer-star allows a $\beta+1$-bend-extension.
\end{reptheorem}

\subsection{Planar star complexity of polygons}
\label{ch:resolve-subsection}

While planar (outer) star complexity nicely bounds the required number of bends per edge in a drawing extension, it does not represent a simple and inherent polygon characteristic. Thus, in the following we first provide an upper bound on the planar (outer) star complexity of a polygon in terms of the size of the polygon (Lemma~\ref{th:ExtendArbitraryCycleToStar}). Then, after  preliminary results, we  provide an upper bound of a planar (outer) star complexity in terms of (outer) star complexity (Theorem~\ref{th:OuterToPlanar}). 

\begin{replemma}{th:ExtendArbitraryCycleToStar}
	A simple polygon with $h$ vertices is a planar $\frac{h-2}{2}$-star and a planar $\frac{h-2}{2}$-outer-star.
\end{replemma}

\begin{proof}
	For the interior, we set a kernel $K$ to be an intersection of the interior of $P$ with an $\eps$-ball around a vertex $u$ of $P$. Let $p$ be a point in $K$. Notice, that by just following the boundary of the polygon it is possible to reach $p$ from any vertex $v\neq u$ with a polyline $c(v)$ with at most $\frac{h-2}{2}$ bends. A set of such curves $\{c(v_i)|v_i \in P\}$, drawn in an appropriate order in order to avoid mutual intersections, represents a planar curve-set of $P$.

	For the exterior, observe that by  following the boundary of the polygon from any vertex $u$ of $P$ it is possible to reach a vertex belonging to the convex hull of $P$ with a polyline $c(u)$ with at most $\frac{h-4}{2}$ bends because the convex hull contains at least three vertices. A set of such curves, drawn in appropriate order in order to avoid mutual crossing, augmented by infinite rays, result in a planar curve-set of $P$ with curve complexity at most $\frac{h-2}{2}$.  \qed
\end{proof}


Observe that, in general, the planar star complexity of a polygon may be  much lower than $\frac{h-2}{2}$. Thus, in the following we aim to bound the  planar star complexity in terms of the star complexity.  
We rely on the following definitions: let $\C$ be a planar curve-set of a planar-$\beta$(-outer)-star.
For a curve $c(v)$ from $\C$ and a point $p$ on  $c(v)$, we denote by $c_v(p)$ the part of the curve split at $p$, not containing $v$ and by  $\#c_v(p)$  the number of bends on $c_v(p)$.
Furthermore, $c(v, p)$ designates the part of the curve $c(v)$ between $v$ and $p$.
An intersection between the curves $c(v)$ and $c(w)$ of $\C$ at a point $p$ is called \emph{avoidable} 
if one of the curves has more bends after the intersection than the other, i.e. if $\#c_v(p)\neq \#c_w(p)$.
The term ``avoidable'' stems from the fact that if $\#c_v(p)>\#c_w(p)$, we can modify $c(v)$ by rerouting it along $c(w)$ starting just before the point $p$ and this way eliminate the intersection without increasing the number of bends per curve.
Concerning said avoidable intersections the following holds:

\begin{replemma}{th:NoAvoidable}
For a given $\beta$(-outer)-star  $P$ there is a curve-set of $P$ with at most $\beta$ bends each without avoidable intersections.
\end{replemma}

In order to resolve all remaining intersections we consider pairs of curves $a$ and $b$ intersecting at a point $p$, such that $p$ is the first intersection for both $a$ and $b$.
In that case we call $p$ \emph{initial} intersection.
However, we first have to show that if there are intersections, then there is always at least one initial intersection.
We formalize this in the following definition and Lemma~\ref{th:NoCyclicOrdering}.
A sequence of vertices $(w_1, \dots, w_m)$ of $P$, with respective curves $(c(w_1), \dots,c(w_m))$ is called \emph{cyclic ordering}, if
for each $1\leq j\leq m$, the first curve that $c(w_j)$ intersects is the curve $c(w_{(j \mod m)+1})$.
We can prove the following:

\begin{replemma}{th:NoCyclicOrdering}
For a given polygon $P$ with a curve-set $\{c(v)\mid v\in V(P)\}$
without avoidable intersections
there is no cyclic ordering.
\end{replemma}


Using Lemma~\ref{th:NoAvoidable} and Lemma~\ref{th:NoCyclicOrdering} we prove a relation between $\beta$-stars and planar-$\beta$-stars.

\begin{theorem}\label{th:OuterToPlanar}
Every $\beta$-star (resp. outer-star) with $n$ vertices is a planar-$(\beta+\delta)$-star (resp. outer-star), where $\delta\leq \log_2(h)$.
\end{theorem}
\begin{proof}
Let $P$  be a $\beta$(-outer)-star with $h$ vertices.
By Lemma~\ref{th:NoAvoidable}, $P$ has a curve-set with at most $\beta$ bends per curve without avoidable intersections.  
Let $p$ be an initial intersection of two curves, which exists by Lemma~\ref{th:NoCyclicOrdering}. We resolve the intersection $p$ by adding a bend to one of the curves and rerouting it along and  sufficiently close to the other  to ensure that they have the same intersections with other curves. We call such curves that follow each other after a resolved intersection a \emph{group}.
We then repeat resolving intersections of groups until there are no more intersections.
As a final part of the proof we show that during this process for each curve at most $\log_2(h)$ bends have been added.

For a curve $c$, let $\bb{c}$ be the number of bends that were added to $c$ during this algorithm.
During the execution of the algorithm we maintain a set of groups $\G$.  Each group $\Gr_i\in\G$ is a set of curves. For each group $\Gr_i$ let $\bb{\Gr_i}$ be the maximum number of additional bends over all curves in $\Gr_i$, i.e.\ $\bb{\Gr_i}=\max_{c\in \Gr_i}(\bb{c})$.
In the beginning each curve is in its own group, that means we start with $\G=\{\{c(v)\}\mid v\in V(P)\}$ and
 for each $\Gr_i\in \G$,  $\bb{\Gr_i}=0$.

The following step is repeated until there are no more intersections.
Let $p$ be an initial intersection of two groups  $\Gr_i$ and  $\Gr_j$.
We  reroute the curves of one of  $\Gr_i$ and  $\Gr_j$. If  we choose to reroute $\Gr_j$, then we add a bend to each curve of $\Gr_j$ and then the curves of $\Gr_j$ follow along the curves of $\Gr_i$, thus increasing $\bb{c}$ by one for each $c\in \Gr_j$. Resolving the intersection $p$ creates a new group $\Gr_k=\Gr_i\cup \Gr_j$.
In order to keep $\bb{\Gr_k}$ bounded we apply the following strategy: if $\bb{\Gr_i}\neq \bb{\Gr_j}$, then we reroute the group with less additional bends and get $\bb{\Gr_k}=\max\{\bb{\Gr_i}, \bb{\Gr_j}\}$. Otherwise, $\bb{\Gr_i}= \bb{\Gr_j}$ and  we arbitrarily choose one of the groups, so $\bb{\Gr_k}=\bb{\Gr_j}+1$.
With each resolved intersection two groups are merged into one, thus the overall number of groups reduces by one. As a result, after at most $h-1$ resolved crossings between groups this iteration stops.

After the above procedure no two curves intersect, thus $P$ is a planar-$\beta\!+\!\delta$-star (resp. outer-star) with $\delta=\max_{\Gr\in\G}(\bb{\Gr})$.
In the following we prove by induction over the group size that for each group $\Gr$ it holds that $\bb{\Gr}\leq\log_2(|\Gr|)$.
For the induction base we observe that if $|\Gr|=1$ we have $\bb{\Gr}=0=\log_2(|\Gr|)$.
As an induction hypothesis, assume that for a $k\geq 1$ and each group $\Gr$ with $|\Gr|\leq k$, it holds that $\bb{\Gr}\leq\log_2(|\Gr|)$ .
Let $\Gr_l$ be a group with $|\Gr_l|=k+1$, which is the result of  merging two groups $\Gr_i$ and $\Gr_j$. Since $|\Gr_i|, |\Gr_j| < |\Gr_l|$,  the induction hypothesis holds for both $\Gr_i$ and $\Gr_j$.
If $\bb{\Gr_i}\neq\bb{\Gr_j}$, we have \(\bb{\Gr_l}=\max\{\bb{\Gr_i}, \bb{\Gr_j}\}
	\leq\log_2(\max\{|\Gr_i|,|\Gr_j|\})
	<\log_2(|\Gr_l|)\). Otherwise, if  $\bb{\Gr_i}=\bb{\Gr_j}$, lets  assume w.l.o.g. $|\Gr_i|\geq|\Gr_j|$, and therefore $|\Gr_l|\geq2|\Gr_j|$.
We have \(\bb{\Gr_l}=\bb{\Gr_j}+1
	\leq\log_2(|\Gr_j|)+1
	\leq\log_2(|\Gr_l|/2)+\log_2(2)
	=\log_2(|\Gr_l|)\).

Since for each $v$ of $P$ the curve $c(v)$ appears in exactly one group, we have that the maximum size of a group is $h$. It follows that $P$ is a planar-$\beta\!+\!\delta$-star (resp. outer-star) with $\delta=\max_{\Gr\in\G}(\bb{\Gr})\leq\log_2(h)$. \qed
\end{proof}

	\section{Drawing extensions of connected subgraphs}\label{ch:GeneralGraphs}
In this section we apply the results from the previous section to provide a tight upper bound on the number of bends in a drawing extension of a connected subgraph.

\begin{reptheorem}{th:ExtendArbitraryCycle}
	Each instance $(G, \Gamma_H)$ where $H$ is an induced connected subgraph of $G$ allows a $\min\{ h/2, \beta + \log_2(h) + 1\}$-bend-extension, where $h$ is the maximum  face size of $H$ and $\beta$ is the maximum (outer) star complexity of a face in $\Gamma_H$. This bound is tight up to an additive constant.
\end{reptheorem}

Above theorem implies an upper bound on the number of bends in case of a non-induced subgraph by simply subdividing the induced edges by dummy vertices and removing them after construction. The tightness of the bound follows from the fact that the lower bound proofs (Theorem~\ref{th:InnerStarAsInnerFaceLowerBound} and Theorem~\ref{th:lowerlog}) can easily be adapted to work for chords.

\begin{corollary}
		Each instance $(G, \Gamma_H)$ where $H$ is a connected subgraph of $G$, allows a $\min\{h +1, 2\beta + 2\log_2(h) + 3\}$-bend-extension, where $h$ is the maximum face size of $H$ and $\beta$ is the maximum star complexity of a face in $\Gamma_H$. This bound is tight up to an additive constant.
\end{corollary}

	\section{Extending stars with straight lines}\label{ch:NoBends}
Let $G=(V, E)$ be a plane graph and $F$ a chordless face, fixed on the plane as a star-shaped polygon $\Gamma_F$.
In this section we study the question whether $(G,\Gamma_F)$ admits a straight-line extension.
Note that for $F$ being the outer face of $G$, Hong and Nagamochi~\cite{hong2008convex} showed that $(G,\Gamma_F)$  always admits a straight-line extension.
In the following $F$ is an inner face.

If $F$ is an inner face fixed as a convex polygon $\Gamma_F$, Mchedlidze~\emph{et al.}~\cite{mchedlidze2016extending} showed that it can easily be tested if an instance $(G,\Gamma_F)$  admits a straight-line extension.
In their case a necessary and sufficient condition for an extension to exist is that for each vertex individually there is a valid position outside $\Gamma_F$.
For stars a comparable result is not possible. Even if each vertex could be drawn individually this does not mean that the whole instance admits a straight-line extension. Even more, testing whether pairs of vertices can be drawn together would not be sufficient as the construction  in Figure~\ref{fig:StarNotAsEasyBigger} suggests.

\begin{figure}[b]
	\centering
	\includegraphics[scale=0.8]{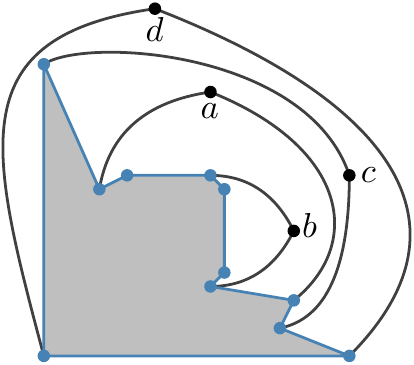}
	\caption{The drawing cannot be extended to a straight-line drawing of the entire graph, even though this is not revealed when testing individual parts.}
	\label{fig:StarNotAsEasyBigger}
\end{figure}

In case of $\Gamma_F$ being a convex inner face (\cite{mchedlidze2016extending}), the feasibility area of a vertex adjacent to the fixed face is just a wedge, formed by the intersection of two half planes induced by two edges of $\Gamma_F$.  In this section we show that the situation for the star shaped inner face is dramatically different, thus there exists an instance for which the feasibility area of a vertex is partially bounded by a curve of exponential complexity.
\begin{reptheorem}{th:ExpoComplBaseCase}
There is an instance $(G,\Gamma_F)$ where $\Gamma_F$ is a star-shaped inner face, such that the feasibility area of some vertex $v\in G$ is partially bounded by a curve whose implicit representation is a polynomial of degree $2^{\Omega(|V|)}$.
\end{reptheorem}

\begin{proofsketch}
A curve is \textit{$i$-exponentially-complex} if it has a parametric representation of the form  $\left\{\left(\frac{r(t)}{u(t)},\frac{s(t)}{u(t)}\right)\mid t\in\I\right\}$ where $r$, $s$ and $u$ are polynomials of degree $2^i$ and $\I$ is an interval.
In the following we describe an instance $(G,\Gamma_F)$, for which a feasibility area of a vertex $v$ is bounded by  an $2^{\Omega(|V|)}$-exponentially-complex curve.
By slightly pertubing the positions of the vertices of $\Gamma_F$ to achieve points in general position we have that the implicit representation of this curve is a polynomial of degree at least $2^{\Omega(|V|)}$.

Let $k\geq 1$ be a fixed integer. Figure~\ref{fig:graphs-complexity-result}a displays the plane graph $G_k=(V,E)$ and the drawing of its inner face as a star-shaped polygon. The vertices $v_i$ and $w_i$ still need to be drawn.
For $0\leq i\leq k$, the feasibility area of $v_i$ is denoted by $A_i$ and the boundary of $A_i$ is referred to as $B_i$. We show that $B_k$ contains a $2^k$-exponentially-complex curve. The proof is by induction on $0\leq i\leq k$.

\begin{figure}[t]
	\centering
	\includegraphics[scale=0.8]{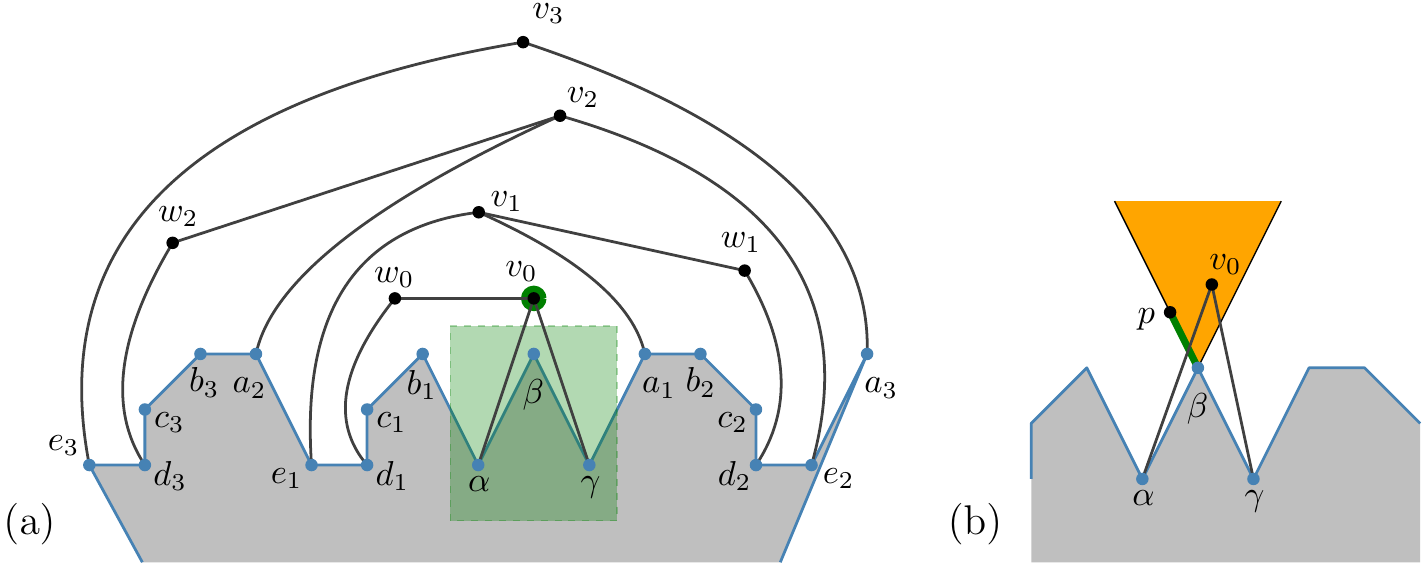}
	\caption{(a) Graph $G_3$, the fixed face is drawn in gray. The vertices in the green area are part of the base case. (b) The base case. The green curve $\C_0$ is on the boundary $B_0$ of the feasibility area $A_0$ of $v_0$.}
	\label{fig:graphs-complexity-result}
\end{figure}

As the base of the induction we consider the boundary $B_0$ of vertex $v_0$ as shown in Figure~\ref{fig:graphs-complexity-result}b.
The feasibility area of $v_0$ is the upper quadrant formed by the lines $l(\alpha, \beta)$ and $l(\gamma, \beta)$.
Let $p=(-1/2, 1)$ be a point on the left boundary of $A_0$. Let $\C_0$ be the segment $s(\beta, p)$ not containing the point $p$.
It holds that $\C_0=\{(\frac{-t}{t+1}, \frac{2t}{t+1})\mid t\in\I=[0, 1)\}$.
The curve $\C_0$ is $0$-exponentially-complex.  An implicit equation of $l(\beta, p)$ is $y+2x=0$.

In the following we assume that the feasibility area $A_{i-1}$ of $v_{i-1}$ is partially bounded by an $i-1$-exponentially-complex curve satisfying additional invariants and prove that the feasibility area $A_i$ of $v_i$ is partially bounded by an $i$-exponentially-complex curve that also satisfies these invariants.
The invariants are given in three groups, the \emph{universal} invariants, holding after each inductive step, the \emph{even} and the \emph{odd} invariants holding after each even and odd step $i\geq 0$, respectively.
Below are the universal and even invariants, with the odd invariants being symmetric.

\pagebreak
\noindent
\textbf{Universal invariants:}
	\begin{enumerate}[label=$\mathcal{UI}$.\arabic*, align=left]
		\item\label{cond:curve} $A_i$ is partially bounded by an $i$-exponentially-complex curve $\C_i=\{v_i(t)=(v_i^x(t), v_i^y(t))\mid t\in \I\}$, where $\I=[0, \I_\max)$ and $\I_\max>0$,
		\item\label{cond:increasing} $v_i^y(t)$ is strictly increasing for $t\in [0, \I_\max)$,
	\end{enumerate}
\noindent
\textbf{Even invariants:}
	\begin{enumerate}[label=$\mathcal{EI}$.\arabic*, align=left]
		\item\label{cond:even-zero-position} $b_{i+1}^x < v_i^x(0) < a_{i+1}^x$ and $v_i^y(0) = 0$,
		\item\label{cond:even-Ai-right} $A_i$ is on the right of $\C_i$,
		\item\label{cond:even-nothing-left} Ray $q(a_{i+1}, v_i(0))$ intersects no point of $A_{i}$ to the left of $v_i(0)$. 
	\end{enumerate}
We observe that universal and even invariants hold for the base case $i=0$.

Let $\C_{i-1}=\left\{v_{i-1}(t) = \left(\frac{r(t)}{u(t)},\frac{s(t)}{u(t)}\right)\mid t\in\I\right\}$.
The position $w_{i-1}(t)$ of vertex $w_{i-1}$ is described as the intersection of the rays $q(v_{i-1}(t), b_i)$ and $q(d_i, c_i)$.
The position $v_i(t)$ of $v_i$ is described as $q(a_i, v_{i-1}(t)) \cap q(e_i, w_{i-1}(t))$. 

Using this we calculate the curve $\C_i$, i.e.\ we calculate the position of $v_i$ as a function of $t$.
This can be done by calculating the equation of the line $l(v_{i-1}(t)),b_i)$, the position of the vertex $w_{i-1}(t)$ and then the equations if the lines $l(a_i,v_{i-1}(t))$ and $l(e_i,w_{i-1}(t))$. The intersection of the latter lines is $v_i(t)$ and we obtain $\C_i=\{(\frac{r_i(t)}{u_i(t)},\frac{s_i(t)}{u_i(t)})\mid t\in \I\}$, where each of $r_i(t)$, $s_i(t)$, $u_i(t)$ is quadratic in $r(t)$, $s(t)$ and $u(t)$.
By induction hypothesis, $C_{i-1}$ is an $i-1$-exponentially complex curve, i.e.\ $u(t)$, $s(t)$, $r(t)$ contain terms $t^{2^{i-1}}$\!. So the curve $\C_i$ is $i$-exponentially complex, provided that the coefficients of highest degree do no cancel themselves out, which can be avoided by slightly perturbing the position of vertex $e_i$. This proves Invariant~\ref{cond:curve}. A proof that the remaining invariants hold after the induction step concludes the proof of the theorem.
\qed
\end{proofsketch}

\section{Conclusion}
\label{ch:conclusion}

We have shown that  a drawing $\Gamma_H$ of an induced connected subgraph $H$ can be extended with at most $\min\{h/2, \beta + \log_2(h) + 1\}$ bends per edge if the star complexity of $\Gamma_H$ is $\beta$ and $h$ is the size of the largest face of $H$ and that this bound is tight up to a small additive constant. In the event of a disconnected subgraph $H$ the known upper bound is $72|V(H)|$. It is tempting to investigate whether the constant $72$ can be lowered and to provide a matching lower bound. 

We have proven that there is an instance $(G,\Gamma_F)$ where $\Gamma_F$ is a star-shaped inner face, such that the feasibility area of some vertex $v\in G$ is partially bounded by an exponential degree curve. 
This is an indication that for a given instance $(G, \Gamma_F)$ it is difficult to test whether $(G,\Gamma)$ admits a straight-line extension. It would be interesting to establish the computational complexity of this problem. We were not able to show the NP-hardness of the problem. Due to its similarity with visibility and stretchability problems we conjecture that the problem is as hard as the existential theory of reals. 

~\\
\noindent
{\bf Acknowledgment}  The authors  thank  Martin N\"ollenburg and Ignaz Rutter for the discussions of this problem back in 2012. J.\,Urhausen has been supported by the Netherlands Organisation for Scientific Research under project 612.001.651.


	\bibliographystyle{abbrv}
	\bibliography{references}

	\appendix
	\label{ch:appendix}

	
	\newpage
	
	\section*{Appendix}
\subsection*{$\beta$-stars}
Figure~\ref{fig:DifferencePlanar} shows an example proving that not each $\beta$-outer-star is also a planar-$\beta$-outer-star. The same holds for (planar-)$\beta$-stars.

\begin{figure}[h]
	\centering
	\includegraphics[scale=0.8]{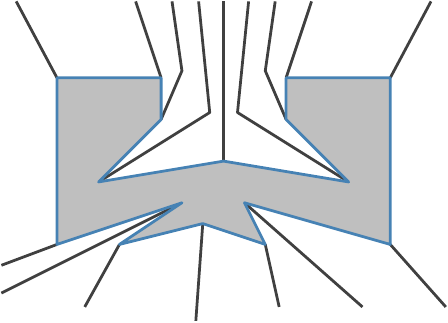}
	\caption{This polygon is a $0$-outer-star, but only a planar-$2$-outer-star.}
	\label{fig:DifferencePlanar}
\end{figure}

\subsection*{Proofs of Subsection~\ref{ch:beta-subsection}}

\repeattheorem{th:InnerStarAsInnerFaceLowerBound}
\begin{proof}

Let $h\geq1$ be a fixed odd integer.
We construct a graph $G_h=(V, E)$ containing a face $F$ with $|F|=h$. We draw $F$ such that the minimal number of bends needed to extend the drawing equals $\ell=(h-3)/2$.
The following construction is visualized in Figure~\ref{fig:k-visible}a.

We set $V = \{a_j, b_j\mid 0\leq j< \ell\}\cup \{a_\ell, u, w, v\}$. Notice that $|V|=h+1$.
In the following we assume that $\ell$ is even. For an odd $\ell$, we use $V\setminus \{a_\ell, b_{\ell-1}\}$ as a vertex set and the analysis determining the number of bends needed is the same.
The face that is fixed is bounded by the following vertices:
\[F= \{ u, a_\ell, a_{\ell-1}, \dots, a_1, a_0, b_0, b_1, \dots, b_{\ell-1}, w \}\]
In addition to the edges induced by $F$, the graph contains the two edges $\{a_0, v\}$ and $\{b_0, v\}$.
The face $F$ is drawn as follows:
\begin{itemize}
	\item $a_j=(3\lfloor j/2 \rfloor, 3\lfloor (j+1)/2 \rfloor), \forall\,1\leq j\leq \ell$,
	\item $b_j=(3\lfloor j/2 \rfloor+1, 3\lfloor (j+1)/2 \rfloor)-1, \forall\,1\leq j< \ell$,
	\item $u=(-3\lfloor (\ell+1)/2 \rfloor-4, 3\lfloor (\ell+1)/2 \rfloor+2)$,
	\item $w=(3\lfloor \ell/2 \rfloor, - 3\lfloor \ell/2 \rfloor -2)$.
\end{itemize}
The resulting polygon $P$ is a $1$-star.
We prove that in order to extend the drawing of $F$ as $P$ to a full drawing of $G_h$, at least one of the edges adjacent to $v$ needs at least $\ell$ bends.
\begin{definition}
We say that a point $q$ outside $P$ is \emph{$k$-visible} from a vertex $z\in F$, if $k$ is the smallest integer such that there exists a k-bend polyline $c(z, q)$ outside $P$ connecting $z$ and $q$ such that $c(z,q)\cap P = \emptyset$.
\end{definition}

\begin{figure}[ht]
	\centering
	\includegraphics[scale=0.8]{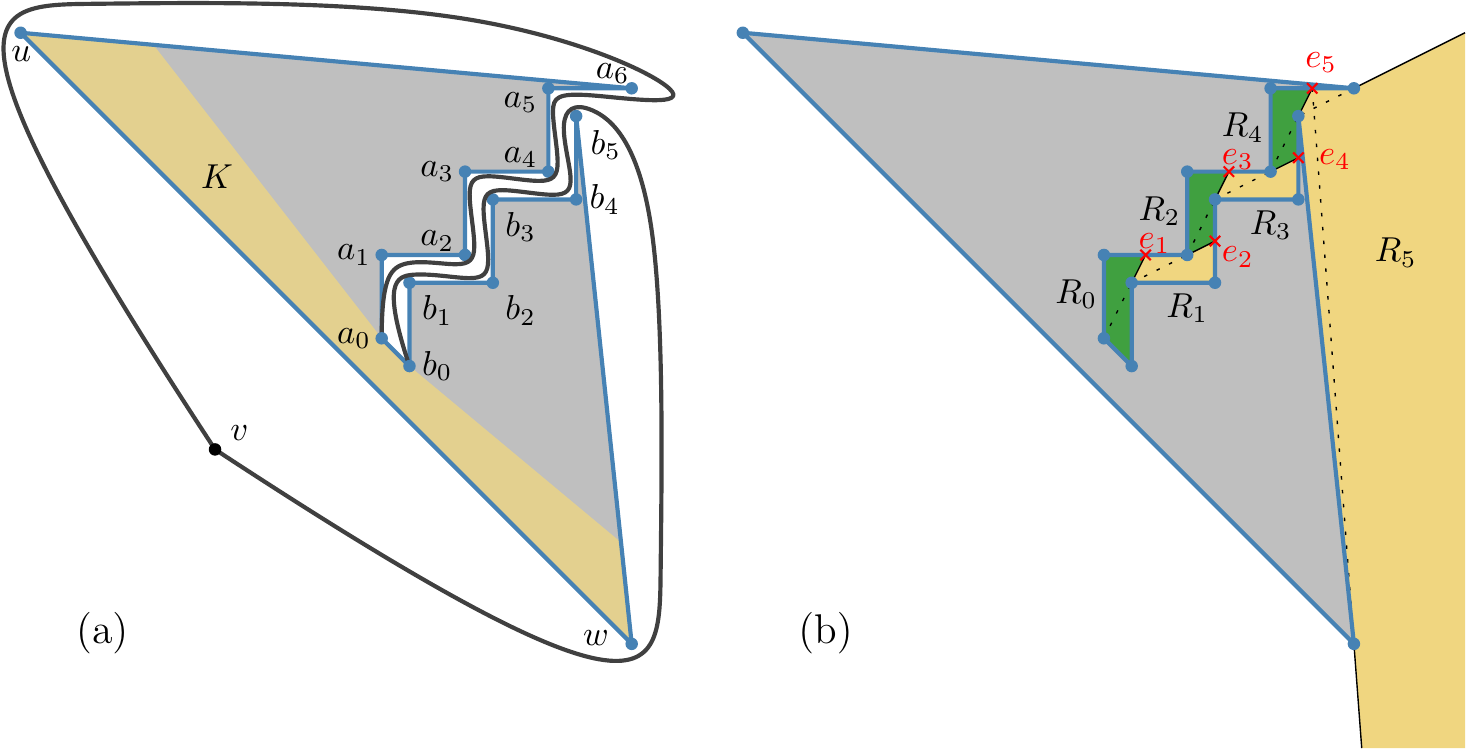}
	\caption{(a) The embedding of $G_2$ along with the fixed drawing of the face $F$ as a $1$-star. The kernel $K$ is beige.
	(b) The regions $R_k$ for $G_2$. The regions with even $k$ are colored green and the odd ones beige.}
	\label{fig:k-visible}
\end{figure}

Let $R_k$ (resp. $R'_k$) be the set of all points that are $k$-visible from $a_0$ (resp. $b_0$).
In order to describe the regions $R_k$ we define the point set $\{e_j\mid 1\leq j < \ell\}$
\begin{align*}
	&\{e_i\}=s(a_i,a_{i+1})\cap q(a_{i-1},b_i)\mbox{, for }i\mbox{ odd,}\\
	&\{e_i\}=s(b_i,b_{i+1})\cap q(b_{i-1},a_i)\mbox{, for }i\mbox{ even.}
\end{align*}

It then follows that for $k\leq \ell-2$ the regions $R_k$ are polygons, see Figure~\ref{fig:k-visible}b.
\begin{align*}
	&R_0=(a_0, a_1, e_1, b_1, b_0)\mbox{,}\\
	&R_i=(e_i, a_{i+1}, e_{i+1},b_{i+1}, b_i)\mbox{, for }i\mbox{ odd,}\\
	&R_i=(a_i, a_{i+1}, e_{i+1}, b_{i+1}, e_i)\mbox{, for }i\mbox{ even.}
\end{align*}
The region $R_{\ell-1}$ has infinite area and is to the right of its boundary consisting of the ray $q(e_{\ell-1}, w)$, the segments induces by the sequence $(w, b_{\ell-1}, e_{\ell-1}, a_\ell)$ and the ray $q(b_{\ell-1}, a_\ell)$.
Additionally we have $R'_0= (a_0, a_1, e'_1, b_0)\subseteq R_0$ with $e'_1=r(b_0, b_1)\cap s(a_1, a_2)$,
$R'_1= R_1 \cup (R_0\setminus R'_0)$ and $R'_i=R_i$ for $i\geq 2$.

From the form of $R_{\ell-1}$ it follows that one cannot draw $v$ such that the two edges $\{v, a_0\}$ and $\{v, b_0\}$ have at most $\ell-1$ bends each. It is however possible with $\ell$ bends each.
\qed
\end{proof}


\repeattheorem{th:lowerlog}
\begin{proof}

For $i\geq 0$ we inductively construct instances $(G_i, \Gamma_{P_i})$.
The instance $(G_0, \Gamma_{P_0} )$  is comprised by the graph $G_0$ that consists of a cycle represented by polygon $P_0$ (see Figure~\ref{fig:LowerBoundExample}a)  and the edges connecting vertices $a$ and $b$ to a vertex~$v$. Vertices $a$ and $b$ are called the \emph{peak} vertices of $P_0$. The outer star complexity of $P_0$ is $0$.
For a peak $w$, let $x$ be the vertex of $P_0$ adjacent to $w$, that is the closest to $w$.
The points outside the polygon whose distance to $w$ is less than half the distance between $x$ and $w$ form the \emph{corridor} of the peak.

\begin{figure}[htbp]
	\centering
	\includegraphics[scale=0.8]{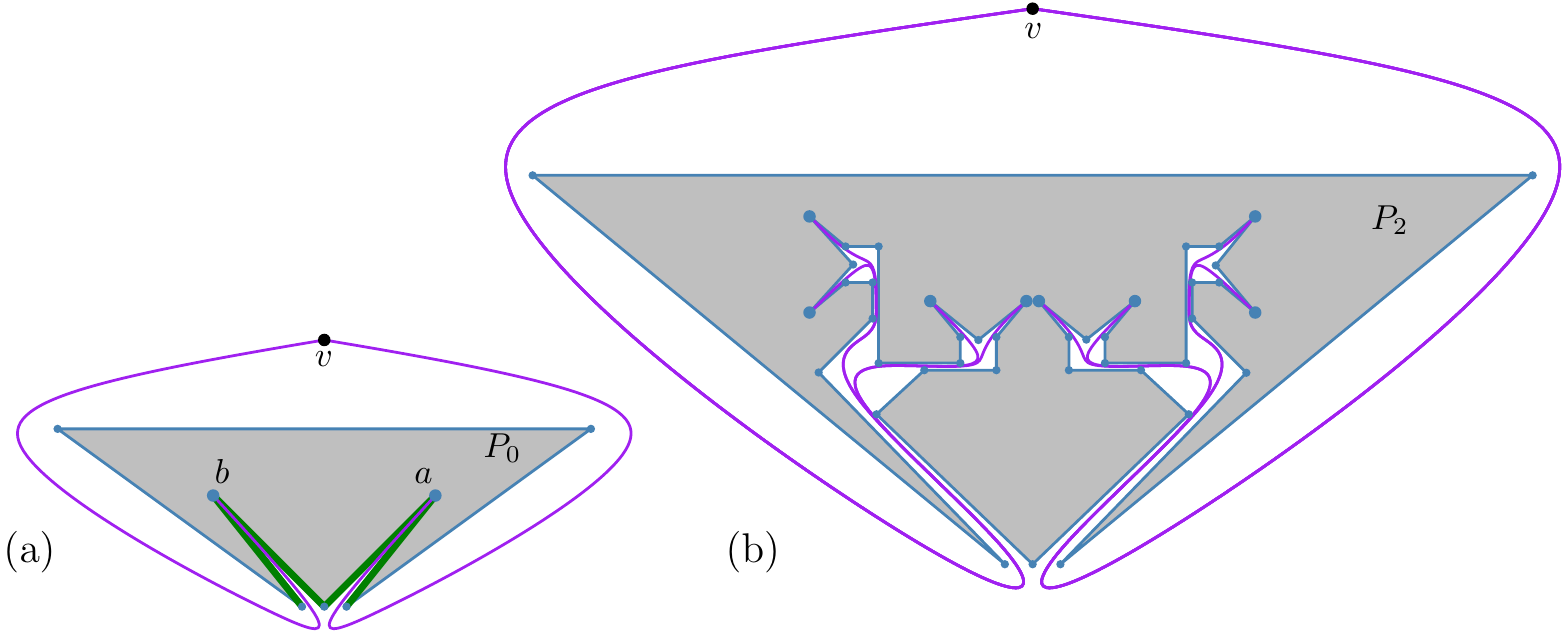}
	\caption{(a)  Instance $(G_0, \Gamma_0)$. (b) Instance $(G_2, \Gamma_2)$. } 
	\label{fig:LowerBoundExample}
\end{figure}

To build the instance $(G_{i+1}, \Gamma_{P_{i+1}})$  from $(G_{i}, \Gamma_{P_{i}})$, we replace each peak $w$ and the edges incident to it by a scaled down \emph{unit} $U$ shown in Figure~\ref{fig:LowerBoundUnit}a. The unit $U$ has two new \emph{peaks} with respective corridors.
We  connect the peak vertices of $U$ to $v$ using the same embedding as the now deleted edge $\{w, v\}$. See Figure~\ref{fig:LowerBoundExample}b for $(G_{2}, \Gamma_{P_{2}})$. 
This step  adds $6\cdot 2^{i+1}$ vertices to the polygon. Thus,  $|V(G_{i+1})|=(6 \cdot 2^{i+1}-5) + 6\cdot 2^{i+1} =6 \cdot 2^{i+2}-5$.
The outer star complexity of $P_{i+1}$ is increased by one compared to $P_i$ and thus it is $i+1$.  

A drawing extension of $(G_i, \Gamma_i)$ with the smallest number  of bends per edge has vertex $v$ positioned  above and in the middle of $P_i$, due to the symmetry of the problem. It is also intuitive that the further away $v$ is, the less bends we need to reach the vertices of $P_i$. So we  assume that $v$ is positioned infinitely far above $v$. In the rest of the proof we prove that the planar outer star complexity of $(G_i,\Gamma_{P_i})$ is $2i+1$ and therefore at least one edge incident to $v$ needs $2i+2$ bends.  
For a fixed $i$ we have that the pre-drawn face has $h= 6 \cdot 2^{i+1}-5$ vertices. This implies $i=\log_2(\frac{h+5}{6})-1$. The polygon $P_i$ has outer star complexity $\beta = i$. So we will get that in any extension of $G_i$  there exists one edge with $2i+2 = \beta+\log_2(\frac{h+5}{6})+1$ bends.

\begin{figure}[ht]
	\centering
	\includegraphics[scale=0.8]{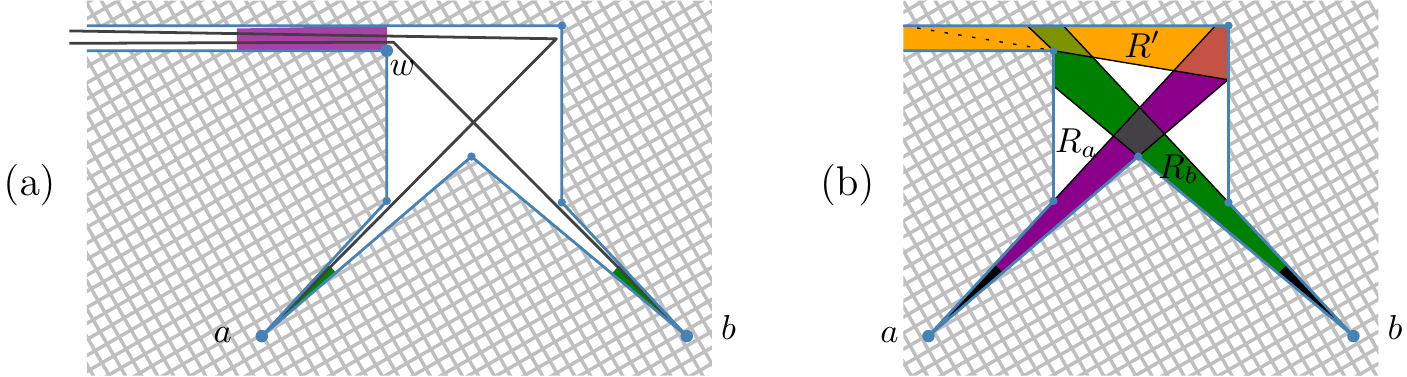}
	\caption{(a)  Illustration of a unit. Vertices $a$ and $b$ are the peeks of the unit.
	The corridors are highlighted.
	(b) Illustration of the regions $R'$, $R_a$ and $R_b$.}
	\label{fig:LowerBoundUnit}
\end{figure}

Our induction hypothesis is that the planar outer star complexity of $(G_i,\Gamma_{P_i})$ is  $2i+1$. Let $\C$ be a planar curve set of $(G_i,\Gamma_{P_i})$. The induction hypothesis implies that for at least one  peak  vertex $w$ its curve $c(w) \in \C$ has $2i+1$ bends.  As an \emph{invariant} we assume that the bend $c$ on $c(w)$ next to $w$ is not in $w$'s corridor. The induction hypothesis and the invariant hold for $i=0$.   Let $U$ denote the unit that has replaced $w$ and let $a, b\in P_{i+1}$ be the peak vertices of $U$.  Let $R'$ be the \emph{visibility region} of $c$ inside $U$, i.e. the set of points inside $U$ that $c$ can see. Assume also that the position of $c$ is such that $R'$ is maximal.

Let $R_a$ and $R_b$ be the visibility regions of $a$ and $b$ inside $P_{i+1}$, respectively, as shown in Figure~\ref{fig:LowerBoundUnit}b.
Observe that due to the shape of the unit and the fact that $c$ is outside the corridor of $w$, $R_a$ and $R_b$ intersect and their intersection lies outside of $R'$.
Thus, any two segments $s(a, a')$ and $s(b, b')$ inside $P_{i+1}$ with $a', b' \in R'$ intersect.
Let us denote by  $p_{a,c}$ (resp. $p_{b,c}$) a polyline inside $P_{i+1}$  with the smallest possible number of bends, connecting $a$ to $c$ ($b$ to $c$, respectively). Any polylines $p_{a,c}$, $p_{b,c}$ with one bend each will intersect. Thus, if  the polylines  $p_{a,c}$, $p_{b,c}$ are non-intersecting, at least one of them has $2$ bends. Since $p_{a,c}$ (resp. $p_{b,c}$) has minimum possible number of bends,  its bend closest to the peak $a$ (resp. $b$), is outside of the corridor of $a$ (resp. $b$).  This implies that the curve complexity of $P_{i+1}$ is $2(i+1)+1$ and that the invariant holds for $(G_{i+1},G_{P_{i+1}})$.
\qed
\end{proof}


\repeattheorem{th:ExtendPlanarInnerStar}
\begin{proof}
We reduce this problem to the problem of extending a partial drawing of the outer face pre-drawn as a star-shaped polygon. The proof is depicted in Figure~\ref{fig:beta-star-to-star}a.
If $\beta=0$, then we just use the algorithm by~\cite{hong2008convex}. Otherwise, let $p$ be a point of the kernel $K$ and let $\C$ be a curve-set of $P$ with center $p$.

Let $v$ be a vertex of $F$, $c(v) \in \C$ be its curve, let $b(v)$ be the last bend on $c(v)$ before $p$ and let $c'(v)=c(v)\setminus s(p,b(v)) \cup \{b(v)\}$ be the curve resulting from removing from $c(v)$ its segment incident to $p$. 
Let $P'$ be the visibility polygon of $p$ inside $P$ restricted by the curves $\cup_{v\in F} c'(v)$.	That is $P'$ is the set of all points $p'$ such that $s(p,p')$ intersects neither $P$ nor a curve $c_i'(v)$.
Clearly we have that the points $\{b(v)| v\in F\}$ are on the boundary of $P'$.

Remember that $N_F(v)=\{u_1,\dots,u_{|N_F(v)|}\}$ is the set of the neighbors of $v$ that are inside $F$.
On each edge $\{v,u_i\}$, $i=1,\dots,|N_F(v)|$  we place a dummy vertex $d_i$. 
We draw the edges $\{v,d_i\}$, $i=1,\dots,|N_F(v)|$ along the curve  $c(v)$ by placing $d_i$ on the boundary of $P'$ very close to $b(v)$.
We connect the dummy vertices with edges as induced by their position on the polygon $P'$, such that they form a cycle. Then each of those edges is drawn following the boundary of $P'$, by placing additional dummy vertices in the event that a bend is needed.
Let $G^*$ be this new graph.
	
Let $G'$ be the subgraph of  $G^*$ induced by  $ V(G^*) \setminus F$.  The vertices on the boundary of $P'$ form the outer face $F'$ of $G'$.
Now we have that $P'$ is a star-shaped polygon and $G'$ has to be drawn inside $P'$. Using the algorithm by~\cite{hong2008convex} we know that $G'$ admits a straight-line drawing inside $P'$.
By merging back the edges of $G$ along the dummy nodes and by removing the edges induced by the dummy vertices, we obtain a drawing extension of $(G,\Gamma_F)$ with at most $\beta$ bends per edge. 
\qed
\end{proof}

\begin{figure}[htbp]
	\centering
	\includegraphics[scale=0.8]{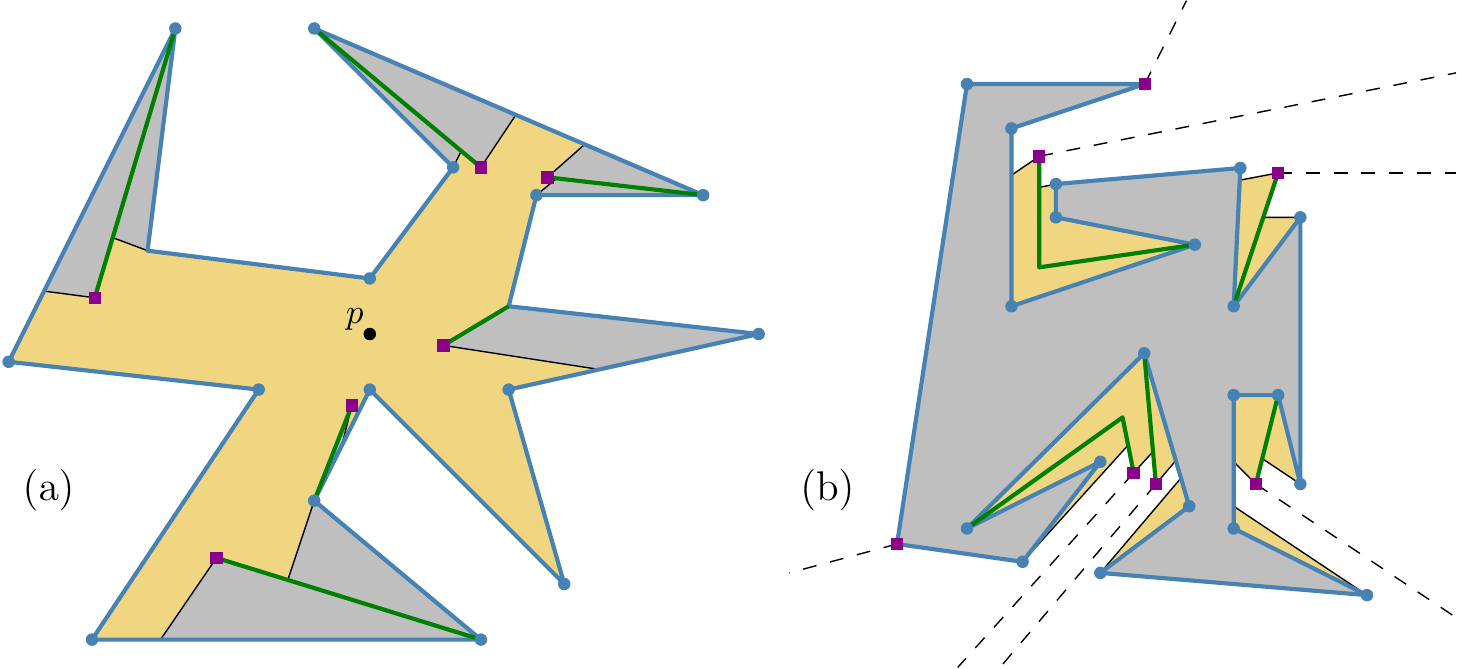}
	\caption{(a) The polygon $P'$ (yellow) is the visibility polygon of $p$. (b) The polygon $P'$ is the union of the yellow and gray regions. The infinite rays of the curves $c_i(v)$ are dashed.}
	\label{fig:beta-star-to-star}
\end{figure}


\repeattheorem{th:ExtendPlanarOuterStar}
\begin{proof}
The proof is along the same line as the proof of Theorem~\ref{th:ExtendPlanarInnerStar} above. It is depicted in Figure~\ref{fig:beta-star-to-star}b.
The only two differences are the following:
First, the center of the curve-set is a point at infinity, so the curves $c(v)$ end with a ray instead of a segment. 

Second, for the polygon $P'$ we cannot simply say that these are the points visible from infinity.
This way $P'$ would only be an outer-star, but not necessarily a planar outer-star.
We need to make sure that there exists rays from the vertices of $P'$ that do not intersect each other.
So let $R$ be the set of the rays of the curves $c_i(v)$. Then let $R'$ be a superset of $R$, such that two rays in $R'$ neither intersect each other nor $P$ nor a curve $c_i(v)$. Additionally, let $R'$ be inclusion maximal with this property.
Then the polygon $P'$ is defined as $\mathbb{R}^2$ minus the points that are on a ray of $R'$.
It follows that $P'$ is a planar outer-star, which by Corollary~\ref{th:corollary-outer-one-bend} implies that we can draw $G'$ outside $P'$ with at most one bend per edge.
\qed
\end{proof}


\subsubsection*{Proofs of Section~\ref{ch:resolve-subsection}}


\repeatlemma{th:NoAvoidable}
\begin{proof}
Let $\C$ be a curve-set of $P$. We assume that the vertices of $P$ are in general position, in the sense that no three curves of $\C$ intersect at the same point and no curve has a bend that lies exactly on another curve.  
If $c(v)$ and $c(w)$ of $\C$ intersect at a point $p$ and we have $\#c_v(p)> \#c_w(p)$, we say that $c(v)$ is \emph{responsible} for this avoidable intersection.
Let $v_1, \dots, v_h$ be an arbitrary ordering of the vertices of $P$. For simplicity of notation, we refer to the curve $c(v_i)$ by $c_i$ and to the subcurve $c_{v_i}(p)$ by $c_i(p)$.
We prove the following statement by induction on $i$: there is no curve $c_j$ with $1\leq j < i$ which is responsible for an avoidable intersection. For $i=1$ this is trivially true.

Now assume that the hypothesis is true for some value of $i<h$.
We eliminate all avoidable intersections that $c_i$ is responsible for. We start with the first avoidable intersection of $c_i$ that it is responsible for, i.e. the intersection of $c_i$ and $c_k$ at $p$ such that $\#c_i(p)> \#c_k(p)$ and such that $c_i$ is not responsible for an avoidable intersection between $v_i$ and $p$.
We reroute $c_i$ by deleting the part $c_i(p)$ and routing $c_i$ along $c_k$.
First, this does not increase the number of bends of $c_i$.
Second, new avoidable intersections on $c_i$ can be created, but these only occur after the point $p$.
The step of resolving the first avoidable intersection that $c_i$ is responsible for is repeated.

Next, we prove that this procedure eliminates at most $\beta$ intersections.
Assume for the sake of contradiction that $c_i$ is responsible for an avoidable intersection with a curve $c_k$ at a point $p$ after having already resolved $\beta$ avoidable intersections that $c_i$ was responsible for. We have that  $\#c(v_i, p)=\beta$ because $c_i$ got a bend at each resolved intersection. We also have that $\#c_i(p)> \#c_k(p)$ because  $c_i$ is responsible for $p$.  We additionally know that $\#c_i\leq \beta$.  It follows that $0\geq \#c_i - \#c(v_i, p)=\#c_i(p)> \#c_k(p)\geq 0$, which is a contradiction.
So each $c_i$ contains at most $\beta$ avoidable intersections. 

We now show that we do not create an avoidable intersection that a $c_j$ with $j<i$ is responsible for.
Assume that  $c_i$ is rerouted along $c_k$ so that $c_i$ now intersects $c_j$ at $p$. It follows that $c_k$ also intersects $c_j$ at a point $p'$ close to $p$. Either the intersection at $p'$ is not avoidable or if it is then $c_j$ is not responsible for it due to the induction hypothesis.
Thus the intersection between $c_i$ and $c_j$ is also not avoidable or if it is, then $c_j$ is not responsible for it.
Combined with the induction hypothesis, it follows that after eliminating all avoidable intersections that $c_i$ is responsible for, we have that for each $1\leq j < i+1$ the curve $c_j$ is not responsible for any avoidable intersection.
\qed
\end{proof}


\repeatlemma{th:NoCyclicOrdering}

\begin{proof}
We assume for the sake of contradiction that a cyclic ordering $(w_1, \dots, w_m)$ exists. We refer to the curve $c(w_j)$ by $c_j$.
Remember that a point $p\in c_j$ splits the curve into the part $c(w_j, p)$ between $w_j$ and $p$ and the remaining part $c_j(p)$.

First, we show that all the intersections of the cyclic ordering happen on a single  segment for each curve. 
For $1\leq i\leq m$, let $j=(i\mod m)+1$ be the next number in cyclic ordering. Let $p_i$ denote the first point of intersection between the curves $c_i$ and $c_j$. 
Since  $p_i$ is not avoidable, both $c_i$ and $c_j$ have the same number of bends after $p_i$, let us denote this number by $z_i$. 
By the definition of cyclic ordering, the intersection $p_j$ appears before the intersection $p_i$ on curve $c_j$, i.e.  $z_j \geq z_i$. Thus $z_{(i\mod m)+1} \geq z_i$ for all $1\leq i \leq m$ and therefore all $z_i$ are equal for $1\leq i\leq m$. Thus, the intersections $p_i$ and $p_j$ lie on the same segment of $c_j$ and we will denote it by $s_j=(e_j, f_j)$.

The union of the curves $c(w_1,p_1)$ and $c(w_2,p_1)$ separates the inside, resp. outside of the polygon into two regions. 
Let $R$ be the region of those two that does not contain the center $k$ of the curve-set.

We next prove that $f_i$ is outside $R$ for each $1\leq i \leq m$. For the sake of contradiction, we assume that $f_i$ is inside $R$; refer to Figure~\ref{fig:NoCyclicOrdering}a for the case  $i=3$. Then the curve $c_i$ has to exit $R$ again because $c_i$ has to reach $k$. 
But $c_i$ can only exit $R$ through the segment $s(p_2,p_1)\subseteq s(e_2, p_1)$ on $c_2$ because $p_1$ and $p_2$ are the first intersections of the curves $c_1$ and $c_2$, respectively.
However, for any intersection $z$ of $c_i$ and $c_2$ on the segment $s(e_2,p_1)$ we have $\#c_2(z)>\#c_i(z)$ because for $c_i$, $z$ follows the bend at $f_i$.
This means that the intersection at $z$ is avoidable.

\begin{figure}[htbp]
\centering
\includegraphics[scale=0.8]{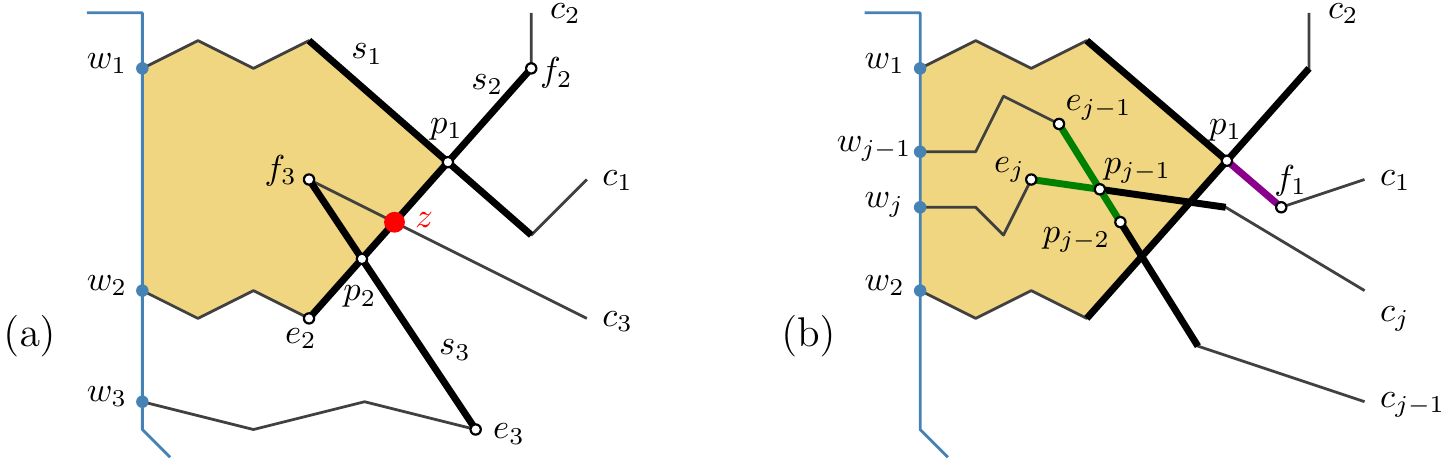}
\caption{The figures accompanying the proof of Lemma \ref{th:NoCyclicOrdering}. The segments $s_j=(e_j, f_j)$ creating the first intersections between the curves in the cyclic ordering are bold. The region R is beige.}
\label{fig:NoCyclicOrdering}
\end{figure}

Now we prove by induction that  segment $s(e_j, p_{j-1})$, $3\leq j\leq m$,  lies inside $R$. 
For $j=3$ we have that the segment $s(e_3, f_3)$ has to intersect the segment $s(e_2,p_1)$ in $p_2$ and $s(e_2,p_1)$  is part of the boundary of $R$. Due to the fact that $f_3$ is outside $R$, $s(e_3, p_2)$ is inside $R$.

We assume as an induction hypothesis that $s(e_j, p_{j-1})$ lies inside $R$, for $j<m$. 
At least some part of the segment $s_{j+1}$ has to be inside $R$ in order to intersect $s(e_j, p_{j-1})$, as can be seen in Figure~\ref{fig:NoCyclicOrdering}b.
As before, the curve $c_{j+1}$ can only intersect the boundary of $R$ at the segment $s(p_2, p_1)$. Additionally, the curve $c_{j+1}$ can only intersect the segment $s(p_2, p_1)$ with the segment $s_{j+1}$, because otherwise an avoidable intersection would exists.
It follows that $f_{j+1}$ has to be outside $R$ and that $e_{j+1}$ has to lie inside $R$. Due to $s(e_j, p_{j-1})$ also being inside $R$, $s(e_{j+1}, p_j)$ is completely inside $R$, too. This finishes the inductive proof.

We have shown that $s(e_m, p_{m-1})$ lies inside $R$. The first intersection of $c_m$ has to belong to the segment $s(e_m, p_{m-1})$ and is with the curve $c_1$ on its segment $s(p_1, f_1)$. However, the segment $s(p_1, f_1)$ lies outside $R$ and therefore cannot intersect $s(e_m, p_{m-1})$. This is a contradiction, which implies that a  cyclic ordering cannot exists. \qed
\end{proof}


\subsection*{Proofs of Section~\ref{ch:GeneralGraphs}}

\repeattheorem{th:ExtendArbitraryCycle}
\begin{proof}
We extend the drawing $\Gamma_H$ to a drawing of the entire graph $G$ face by face.
If $H$ is biconnected, thus each face $F$ is drawn as a simple connected polygon $\Gamma_F$, the upper bound follows trivially by 
observing that $\Gamma_F$ is  a  planar-$\frac{h-2}{2}$-(outer) star (by Lemma~\ref{th:ExtendArbitraryCycleToStar}) and
a planar-$\beta + \log_2(h)$-(outer)-star (by Theorem~\ref{th:OuterToPlanar}), where $h$ is the number of vertices in $F$. 
If $F$ is the outer face of $G$, we can then extend the drawing $\Gamma_F$ using at most $\min\{ \frac{h-2}{2}, \beta + \log_2(h)\}$ bends per edge (Theorem~\ref{th:ExtendPlanarInnerStar}), otherwise we get $\min\{ \frac{h}{2}, \beta + \log_2(h)+1\}$ bends per edge (Theorem~\ref{th:ExtendPlanarOuterStar}).

If $H$ is not biconnected, then $\Gamma_G$ contains faces with multiple occurrences of a vertex. We observe that the techniques applied in the proofs of the aforementioned theorems directly generalize to this case. 

By Theorem~\ref{th:InnerStarAsInnerFaceLowerBound} we know that there are instances where each extension contains an edge with at least $\lfloor\frac{h-3}{2}\rfloor$ instances.
We have $\frac{h}{2} - \lfloor\frac{h-3}{2}\rfloor \leq 2$, for each $h$.
By Theorem~\ref{th:lowerlog} we know that there are instances where each extension contains an edge with at least $\beta+ \log_2(\frac{h+5}{6})+1$ bends.
We have $\beta + \log_2(h)+1 - (\beta+ \log_2(\frac{h+5}{6})+1)  < \log_2(6) \approx 2.58$, for each $h$. So the proven upper bound is tight up to an additive constant of $3$.
\qed
\end{proof}


\subsection*{Proofs of Section~\ref{ch:NoBends}}

\repeattheorem{th:ExpoComplBaseCase}
\begin{proof}
A curve is \textit{$i$-exponentially-complex} if it has a parametric representation of the form  $\left\{\left(\frac{r(t)}{u(t)},\frac{s(t)}{u(t)}\right)\mid t\in\I\right\}$ where $r$, $s$ and $u$ are polynomials of degree $2^i$ and $\I$ is an interval.

In the following we describe an instance $(G,\Gamma_F)$ for which a feasibility area of a vertex $v$ is bounded by  an $2^{\Omega(|V|)}$-exponentially-complex curve. By slightly pertubing the positions of the vertices of $\Gamma_F$ to achieve points in general position, we have that the implicit representation of this curve is a polynomial of degree at least $2^{\Omega(|V|)}$.

\begin{figure}[h]
	\centering
	\includegraphics[scale=0.8]{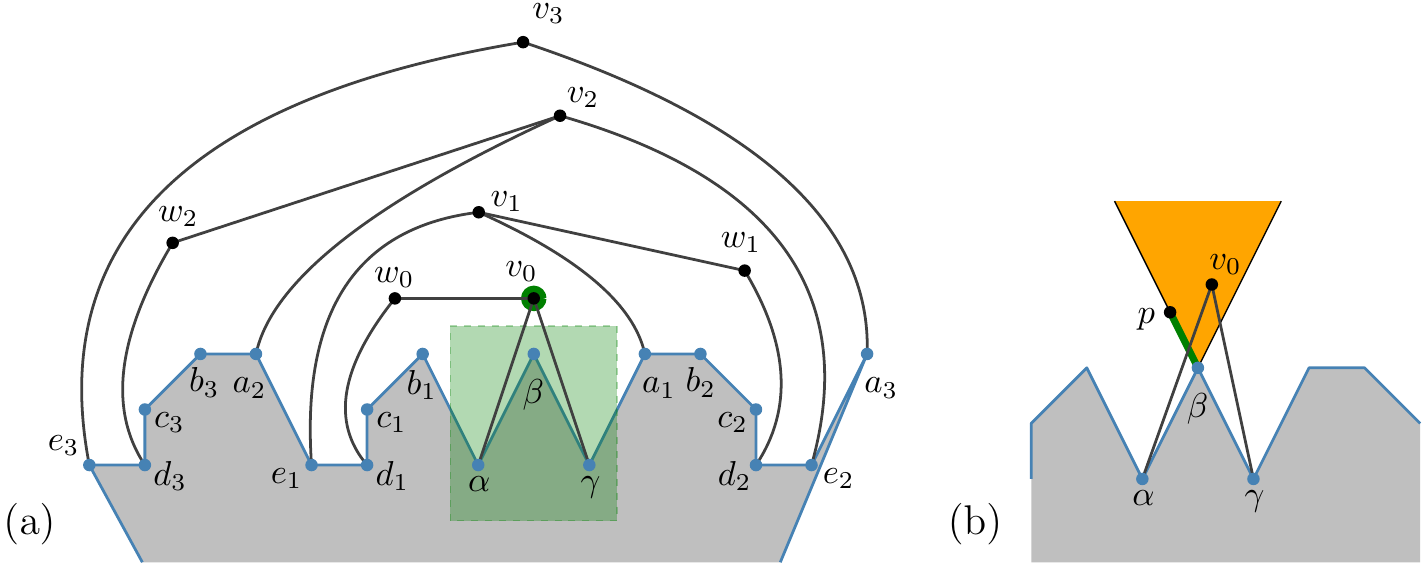}
	\caption{(a) Graph $G_3$, the fixed face is drawn in gray. The vertices in the green area are part of the base case. (b) The base case. The green curve $\C_0$ is on the boundary $B_0$ of the feasibility area $A_0$ of $v_0$.}
	\label{fig:graphs-complexity-result-appendix}
\end{figure}

Let $k\geq 1$ be a fixed integer. In the following we define an instance of size $O(k)$ with the above-mentioned property.
Let $G_k=(V,E)$ be the plane graph shown in Figure~\ref{fig:graphs-complexity-result-appendix}a. The vertex set is $V=\{a_i, b_i, c_i, d_i, e_i\mid 1\leq i\leq k\}\cup\{v_i\mid 0\leq i\leq k\}\cup\{w_i\mid 0\leq i< k\}\cup\{\alpha, \beta, \gamma, h\}$.
The edges are defined as follows: the face $F$ is
\begin{alignat*}{3}
k&=2j: &&a_{2j}, e_{2j-1}, d_{2j-1}, c_{2j-1}, b_{2j-1}, a_{2j-2}, \dots, c_1, b_1, \alpha, \beta, \gamma,\\[-1\jot]
&&&a_1, b_2, c_2, d_2, e_2, a_3, \dots, d_{2j}, e_{2j}, h\\
k&=2j+1: \\[-1\jot]
&e_{2j+1}, \dots, b_{2j+1},\ &&a_{2j},\ e_{2j-1}, d_{2j-1}, c_{2j-1}, b_{2j-1}, a_{2j-2}, \dots, c_1, b_1, \alpha, \beta, \gamma,\\[-1\jot]
&&&a_1, b_2, c_2, d_2, e_2, a_3, \dots, d_{2j}, e_{2j}, a_{2j+1}, h
\end{alignat*}
In addition to the edges around $F$, the edge set contains the following edges for $1\leq i\leq k$:
$\{v_0, \alpha\}, \{v_0, \gamma\},
\{v_i,a_i\}, \{v_i,e_i\}, \{w_{i-1}, d_i\}, \{w_{i-1},v_{i-1}\}$.

For $j\in\mathbb{N}$, the position of the vertices of $F$ is as follows:
\begin{align*}
&\alpha = (-1, -2);
\beta = (0, 0);
\gamma = (1, -2);\\
&a_{2j-1} = (4 j-2,0);
b_{2j-1} = (2-4 j,0);\\
&\quad c_{2j-1} = (1-4 j,-1);
d_{2j-1} = (1-4 j,-2);
e_{2j-1} = (-4 j,-2+\lambda_{2j+1});\\
&a_{2j} = (-1-4 j,0);
b_{2j} = (4 j-1,0);\\
&\quad c_{2j} = (4 j,-1);
d_{2j} = (4 j,-2);
e_{2j} = (4 j+1,-2+\lambda_{2j}).
\end{align*}
The variables $0\leq \lambda_i\ll 1$ are determined later. The vertex $h$ (not shown in the figure) is placed at $(0, h_y)$ for $h_y<0$ and $|h_y|$ big enough so that the polygon formed by $F$ is star-shaped.
For $0\leq i\leq k$, the feasibility area of $v_i$ is denoted by $A_i$ and the boundary of $A_i$ is referred to as $B_i$. We show that $B_k$ contains a $2^k$-exponentially-complex curve. The proof is by induction on $0\leq i\leq k$.

As the base of the induction we consider the boundary $B_0$ of vertex $v_0$ as shown in Figure~\ref{fig:graphs-complexity-result-appendix}b.
The feasibility area of $v_0$ is the upper quadrant formed by the lines $l(\alpha, \beta)$ and $l(\gamma, \beta)$.
Let $p=(-1/2, 1)$ be a point on the left boundary of $A_0$. Let $\C_0$ be the segment $s(\beta, p)$ not containing the point $p$.
It holds that $\C_0=\{(\frac{-t}{t+1}, \frac{2t}{t+1})\mid t\in\I=[0, 1)\}$.
The curve $\C_0$ is $0$-exponentially-complex.  An implicit equation of $l(\beta, p)$ is $y+2x=0$.
In the following we assume that the feasibility area $A_{i-1}$ of $v_{i-1}$ is partially bounded by an $(i\!-\!1)$-exponentially-complex curve satisfying additional invariants and prove that the feasibility area $A_i$ of $v_i$ is partially bounded by an $i$-exponentially-complex curve that also satisfies these invariants.
The invariants are given in three groups, the \emph{universal} invariants, holding after each inductive step, the \emph{even} and the \emph{odd} invariants holding after each even and odd step $i\geq 0$, respectively. 


~\\
\textbf{Universal invariants:}
	\begin{enumerate}[label=$\mathcal{UI}$.\arabic*, align=left]
		\item\label{cond:curve-appendix} $A_i$ is partially bounded by an $i$-exponentially-complex curve $\C_i=\{v_i(t)=(v_i^x(t), v_i^y(t))\mid t\in \I\}$, where $\I=[0, \I_\max)$ and $\I_\max>0$,
		\item\label{cond:increasing-appendix} $v_i^y(t)$ is strictly increasing for $t\in [0, \I_\max)$,
	\end{enumerate}
\noindent
\textbf{Odd invariants:}
	\begin{enumerate}[label=$\mathcal{OI}$.\arabic*, align=left]
		\item\label{cond:odd-zero-position-appendix} $a_{i+1}^x < v_i^x(0) < b_{i+1}^x$ and $v_i^y(0) = 0$,
		\item\label{cond:odd-Ai-left-appendix} $A_i$ is on the left of $\C_i$, i.e.\ $\forall t\in \I\setminus\{0\}, \exists \theta>0:$\\$(v_i^x(t)-\theta,v_i^y(t))\in A_i \land (v_i^x(t)+\theta,v_i^y(t))\notin A_i$,
		\item\label{cond:odd-nothing-right-appendix} Ray $q(a_{i+1}, v_i(0))$ intersects no point of $A_{i}$ to the right of $v_i(0)$, 
	\end{enumerate}

\noindent
\textbf{Even invariants:}
	\begin{enumerate}[label=$\mathcal{EI}$.\arabic*, align=left]
		\item\label{cond:even-zero-position-appendix} $b_{i+1}^x < v_i^x(0) < a_{i+1}^x$ and $v_i^y(0) = 0$,
		\item\label{cond:even-Ai-right-appendix} $A_i$ is on the right of $\C_i$,
		\item\label{cond:even-nothing-left-appendix} Ray $q(a_{i+1}, v_i(0))$ intersects no point of $A_{i}$ to the left of $v_i(0)$. 
	\end{enumerate}
We observe that universal and even invariants hold for the base case $i=0$. In the following we assume that $i-1$ is odd, the case where $i-1$ is even is symmetric.

\begin{figure}[ht]
	\centering
	\includegraphics[scale=0.8]{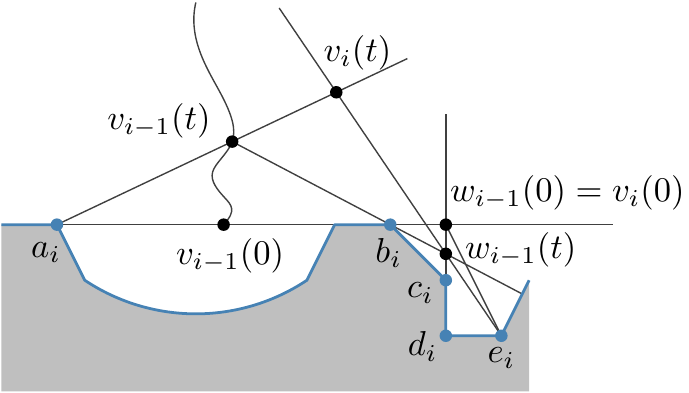}
	\caption{For $t>0$, the position of vertex $v_i(t)$ is above and to the left of $v_i(0)$.}
	\label{fig:recursive-proof-increasing-appendix}
\end{figure}

Let $\C_{i-1}=\left\{v_{i-1}(t) = \left(\frac{r(t)}{u(t)},\frac{s(t)}{u(t)}\right)\mid t\in\I\right\}$.
The position $w_{i-1}(t)$ of vertex $w_{i-1}$ is described as the intersection of the rays $q(v_{i-1}(t), b_i)$ and $q(d_i, c_i)$.
The position $v_i(t)$ of $v_i$ is described as $q(a_i, v_{i-1}(t)) \cap q(e_i, w_{i-1}(t))$. 

In the following we calculate the curve $\C_i$. We therefore calculate the position of $v_i$ as a function of $t$. Recall that the $y$-coordinate of point $e_i$ depends on variable $\lambda_i$. For the sake of simplicity we use $\lambda$ instead of $\lambda_i$ in the following paragraph. The calculations were made using Mathematica~\cite{mathematica}.
\begin{itemize}
	\item The equation of the line $l(v_{i-1}(t)),b_i)$ as a function of $x$ is: $y = \frac{-4 i s(t) +x s(t)}{r(t) - 4 i u(t)}$
	\item The position of $w_{i-1}$ as a function of $t$ is: $w_{i-1}(t)=\left(1 + 4 i, \frac{s(t)}{r(t) - 4 i u(t)}\right)$
	\item Line $l(a_i,v_{i-1}(t))$ as a function of $y$ is $ x = 2 - 4 i + y\frac{r(t) + 2 (-1 + 2 i) u(t)}{s(t)}$
	\item Line $l(e_i,w_{i-1}(t))$ as a function of $y$ is $x = \frac{(-2 + \lambda - 8 i + 4 \lambda i + y) (r(t) - 4 i u(t)) - 2 (1 + 2 i) s(t)}{(-2 + \lambda) r(t) - s(t) - 4 (-2 + \lambda) i u(t)}$
	\item Thus, the intersection of the lines $l(e_i,w_{i-1}(t))$ and $l(a_i,v_{i-1}(t))$ is described by the curve $\C_i=\{(\frac{r_i(t)}{u_i(t)},\frac{s_i(t)}{u_i(t)})\mid t\in \I\}$, where\\
\begin{itemize}
	\item $ r_i(t)= $ \small $-(-2 + \lambda) (1 + 4 i)$ \normalsize $ r(t)^2 + $ \small $4 (-1 - 2 i + 8 i^2)$ \normalsize $ s(t)u(t) + $ \small $8 (-1 - 2 i + 8 i^2) (-2 + \lambda) i$ \normalsize $ u^2(t) + $ \small $4$ \normalsize $ r(t) s(t)   + $ \small $  2 (-2 + \lambda) (1 + 4 i)$ \normalsize $ r(t)u(t)$,
	\item $s_i(t)=$ \small $-(-2 + \lambda) (-1 + 8 i) $ \normalsize $r(t)s(t)  + $ \small $8 i $ \normalsize $s^2(t) + $ \small $4 (-2 + \lambda) i (-1 + 8 i) $ \normalsize $s(t) u(t)$,
	\item $u_i(t)=$ \small $-(-2 + \lambda) $ \normalsize $r(t)^2 + $ \small $2 $ \normalsize $r(t) s(t) +$ \small $ 2(-2 + \lambda)  $ \normalsize $r(t)u(t) -$ \small $ 2 $ \normalsize $s(t)u(t) +$ \small $ 8 (-2 + \lambda) i (-1 + 2 i) $ \normalsize $u^2(t)$.
\end{itemize}

\end{itemize}
First, we observe that the denominators of $v_i^x(t)$ and $v_i^y(t)$ are equal. Second, each of $r_i(t)$, $s_i(t)$, $u_i(t)$ is quadratic in $r(t)$, $s(t)$ and $u(t)$.
Since by induction hypothesis, $C_{i-1}$ is an $i\!-\!1$-exponentially complex curve, i.e. $u(t)$, $s(t)$, $r(t)$ contain terms $t^{2^{i-1}}$\!, the curve $\C_i$ is $i$-exponentially complex, provided that the coefficients of highest degree do no cancel themselves out. 
Let therefore $\bar r$, $\bar s$ and $\bar u$ be the coefficients of $t^{2^{i-1}}$ in $r(t)$, $s(t)$ and $u(t)$.
The coefficients of $t^{2^{i}}$ cancel themselves out in the formula of $s_i(t)$ above if and only if
\small $-(-2 + \lambda) (-1 + 8 i) $ \normalsize $\bar r\bar s t^{2^i}  + $ \small $8 i $ \normalsize $\bar s^2 t^{2^i} + $ \small $4 (-2 + \lambda) i (-1 + 8 i) $ \normalsize $\bar s\bar u t^{2^i}=0$,
for each $t\in \I$, so we fix $t$ and solve for $\lambda$. This equation is linear in $\lambda$, so it only has one root. The same holds if we create an equation for $r(t)$ and $u(t)$.
So the coefficients of $t^{2^i}$ cancel themselves out if and only if
$\lambda$ equals one of the three roots.
We can chose $\lambda$ close to $0$, such that none of these equalities holds, which means that the coefficients of highest degree are non-zero. So Invariant~\ref{cond:curve} holds as long as $\C_i$ is part of $B_i$, which is proven later. In the following we prove that the invariants hold.

We first prove \underline{Invariant~\ref{cond:even-zero-position-appendix}}. By Invariant~\ref{cond:odd-zero-position-appendix}, we have $a_i^x<v_{i-1}^x(0)<b_i^x$ and $v_{i-1}^y(0)=0$. So the ray $q(a_i, v_{i-1}(0))$ is horizontal and points to the right.
The ray $q(d_i, c_i)$ is vertical and points upward and we have $b_i^x< c_i^x$. The point $w_{i-1}(0)$ is the intersection of those rays and therefore has coordinates $(c_i^x,0)$.
This shows that $w_{i-1}(0)\in q(a_i, v_{i-1}(0))$.
The point $v_i(t)$ is defined as $q(a_i, v_{i-1}(t)) \cap q(e_i, w_{i-1}(t))$, so we have that $w_{i-1}(0)=v_i(0)$, which concludes the proof of Invariant~\ref{cond:even-zero-position-appendix}.


In the following we prove \underline{Invariant~\ref{cond:increasing-appendix}} as shown in Figure~\ref{fig:recursive-proof-increasing-appendix}.
Recall that $a_i^y=b_i^y=0$. By Invariant~\ref{cond:odd-zero-position-appendix}, $v_{i-1}^y(0)=0$. Thus, by Invariant~\ref{cond:increasing-appendix}, $v_{i-1}^y(t)>0$ for any $t>0$. Also, by Invariant~\ref{cond:odd-zero-position-appendix}, $a_i^x < v_{i-1}^x(0) < b_i^x$.
Let $\I^{(1)}$ be the maximal half-open subinterval of $\I$ such that $0\in \I^{(1)}$ and $a_i^x < v_{i-1}^x(t) < b_i^x$ for each $t\in\I^{(1)}$. Interval $\I^{(1)}$ is not degenerated because $\C_{i-1}$ is continuous.
Thus the line $l(a_i,v_{i-1}(t))$ has a positive slope and line $l(v_{i-1}(t)),b_i)$ has a negative slope for $t\in \I^{(1)}\setminus\{0\}$. 

Recall that $\{w_{i-1}(t)\}=q(v_{i-1}(t), b_i)\cap q(d_i, c_i)$ and that $b_i^y=0$. Thus, $w_{i-1}^x(t)=c_i^x$ for any $t\in \I^{(1)}$;  $w_{i-1}^y(t)<0$ for any $t\in \I^{(1)}\setminus\{0\}$; and $w_{i-1}^y(0)=0$.
Recall that $\{v_i(t)\}=q(a_i, v_{i-1}(t)) \cap q(e_i, w_{i-1}(t))$. Due to $w_{i-1}(0)\in q(a_i, v_{i-1}(0))$, we have $v_i(0)=w_{i-1}(0)$. 
Therefore, $l(e_i,w_{i-1}(t)) \cap l(a_i,v_{i-1}(t))$ is to the left and above $v_{i}(0)$ for $t\in \I^{(1)}\setminus\{0\}$. 
Since $\C_i$ is continuous, there exists a half-open subinterval $\I^{(2)}\subseteq \I^{(1)}$ such that for $t\in \I^{(2)}$, $v_i^y(t)$ is strictly increasing, proving Invariant~\ref{cond:increasing-appendix}.


In the following we prove that the intersection $v_i(t)$ is on the boundary $B_i$ of $A_i$, thus proving \underline{Invariant~\ref{cond:curve-appendix}}.
We need the following claims:

\begin{claim}\label{claim:subinterval-no-other-intersection-appendix}
For some half open interval $\I^{(3)}\subseteq\I^{(2)}$ we have that for all $t\in \I^{(3)}$, $q(a_i, v_{i-1}(t))\cap A_i=\{v_{i-1}(t)\}$.
\end{claim}
The proof is visualized in Figure~\ref{fig:both-claims}a.
We know that the vertex $a_i$ is left of $v_{i-1}(0)$ by Invariant~\ref{cond:odd-zero-position-appendix} and  that $A_{i-1}$ is on the left of $\C_{i-1}$ by Invariant~\ref{cond:odd-Ai-left-appendix}.
Additionally, we have that $q(a_i, v_{i-1}(0))$ intersects $B_{i-1}$ at $v_{i-1}(0)$ and intersects no point of $A_{i-1}$ right of $v_{i-1}(0)$ due to Invariant~\ref{cond:odd-nothing-right-appendix}.
We also have that $v_{i-1}^y(t)$ is strictly increasing by Invariant~\ref{cond:increasing-appendix} proven above.
Now, due to $\C_i$ being continuous the claim holds.

\begin{figure}[ht]
	\centering
	\includegraphics[scale=0.8]{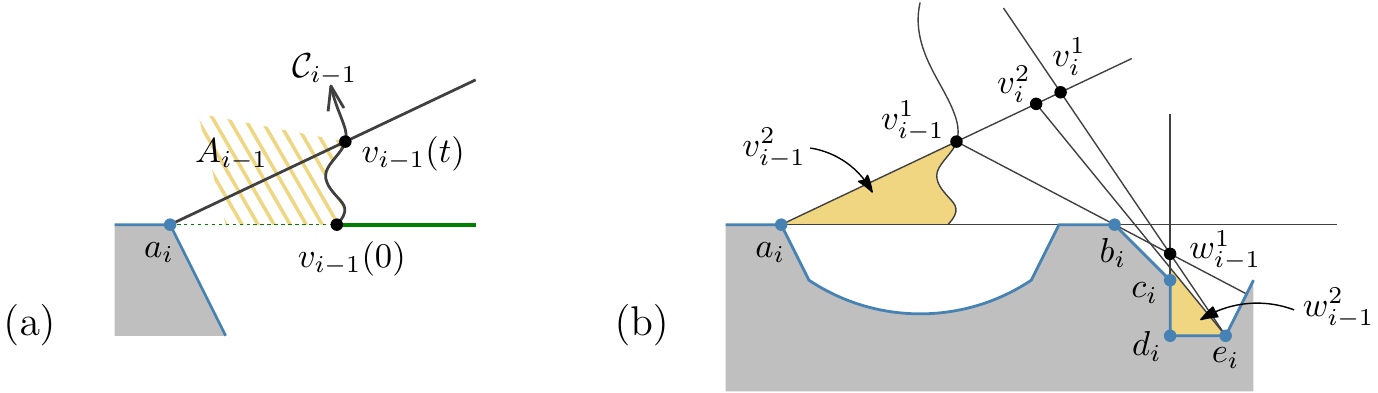}
	\caption{(a) The proof for Claim~\ref{claim:subinterval-no-other-intersection-appendix}. The region $A_{i-1}$ is left of $C_{i-1}$ and the green ray $q(a_i,v_{i-1}(0))$ does not intersect $A_{i-1}$. This means that for small $t$, $q(a_i,v_{i-1}(t))$ does not intersect $A_{i-1}$ either.
	(b) The construction used to prove Claim~\ref{claim:segment-empty-appendix}. We have $v_i^2\in s(v_{i-1}^1, v_i^1)$ and then $v_{i-1}^2$ is in the beige area on the left and $w_{i-1}^2$ is in the beige area on the right, thus they can't be connected by an edge.}
	\label{fig:both-claims}
\end{figure}

Building on that result we prove the following:
\begin{claim}\label{claim:segment-empty-appendix}
For $t\in \I^{(3)}$ we have that $s(v_{i-1}(t), v_i(t))\cap A_i=\{v_i(t)\}$.
\end{claim}
The proof of the claim is visualized in Figure~\ref{fig:both-claims}b.
We already know that $v_i(t)\in A_i$ for $t\in \I$.
In the following we denote the position of the vertices $v_{i-1}(t)$, $w_{i-1}(t)$ and $v_i(t)$ by $v_{i-1}^1$, $w_{i-1}^1$ and $v_i^1$.
We assume for the sake of contradiction that there exists a position $v_i^2\in s(v_{i-1}^1, v_i^1)\setminus \{v_i^1\}$ of vertex $v_i$, such that $v_i^2$ is in $A_i$.
We try to find valid positions $v_{i-1}^2$ and $w_{i-1}^2$ for the vertices $v_{i-1}$ and $w_{i-1}$.
The position $w_{i-1}^2$ is to the right of $l(c_i, d_i)$ and below $l(v_i^2, e_i)$. So $w_{i-1}^2$ is strictly below the line $l(v_{i-1}^1, w_{i-1}^1)$.
The position $v_{i-1}^2$ is below the line $l(a_i, v_i^2)$ and to the left of $\C_{i-1}$ by Claim~\ref{claim:subinterval-no-other-intersection-appendix}. This means that $v_{i-1}^2$ is below the line $l(v_{i-1}^1, w_{i-1}^1)$.
Due to $b_i \in l(v_{i-1}^1, w_{i-1}^1)$, we have that $b_i$ is strictly above $l(v_{i-1}^2, w_{i-1}^2)$, which is a contradiction because the vertices $v_{i-1}$ and $w_{i-1}$ are connected by an edge. This proves the claim.

Now we can quickly show that $v_i(t)\in B_i$ for $t\in\I^{(3)}$, therefore proving Invariant~\ref{cond:curve-appendix}.
We assume for the sake of contradiction that $v_i(t)\notin B_i$.
Then there is a region $R$ around $v_i(t)$ with $R\subseteq A_i$. It follows that $s(v_{i-1}(t), v_i(t))\cap R\setminus\{v_i(t)\}\neq \emptyset$, which is a contradiction to Claim~\ref{claim:segment-empty-appendix}.

We prove \underline{Invariant~\ref{cond:even-nothing-left-appendix}}.
We know that $v_i(0)=(c_i^x, 0)$. Using Invariant~\ref{cond:odd-zero-position-appendix} and the coordinates of the vertices of $F$, we have that $v_{i-1}^x(0)<b_i^x< v_i^x(0)<a_{i+1}^x$ and that points $v_{i-1}(0)$, $b_i$, $v_i(0)$ and $a_{i+1}$ are on the x-axis.
So $v_{i-1}(0)\in q(a_{i+1}, v_i(0))$.
Using $s(v_{i-1}(t), v_i(t))\cap A_i=\{v_i(t)\}$ from Claim~\ref{claim:segment-empty-appendix} we have
\linebreak
$q(a_{i+1}, v_i(0))\cap A_i\setminus \{v_i(0)\}=q(a_{i+1}, v_{i-1}(0))\cap A_i$. However, for $v_i^y=0$, $v_i$ is on the right of the line $l(e_i, b_i)$, which means that $q(a_{i+1}, v_{i-1}(0))\cap A_i=\emptyset$. This proves Invariant~\ref{cond:even-nothing-left-appendix}.

Lastly we prove \underline{Invariant~\ref{cond:even-Ai-right-appendix}}.
Due to Claim~\ref{claim:segment-empty-appendix}, we know that for each $t\in\I^{(3)}$, we have $s(v_{i-1}(t), v_i(t))\cap A_i=\{v_i(t)\}$.
Additionally, we know that for a given position of $v_{i-1}$ and $w_{i-1}$, the vertex $v_i$ can be placed in the upper right quadrant formed by the lines $l(a_i,v_{i-1})$ and $l(e_i,w_{i-1})$. This means that the ray $q(v_{i-1}(t), v_i(t))$ is part of $A_i$.
For $t\in\I^{(3)}$, we have $a_i^x < v_{i-1}^x(t)$ due to the definition of $\I^{(1)}$ and the fact that $\I^{(3)}\subseteq\I^{(1)}$. Since $v_i(t)\in q(a_i, v_{i-1}(t))$, it holds that $a_i^x < v_i^x(t)$.
It follows that the region $A_i$ is right of the curve $\C_i$.
\qed
\end{proof}

\end{document}